\documentclass[onecolumn,letterpaper,journal]{IEEEtran}
\usepackage{upgreek}
\usepackage{mathtools, xparse}
\usepackage{graphicx}
\usepackage{epstopdf}
\usepackage{bm}
\usepackage{scalerel}
\usepackage{amsfonts,amsmath}
\usepackage{amsthm}
\usepackage{algorithm,algorithmic}
\usepackage{mathtools}
\usepackage{multirow}
\usepackage{verbatim}
\usepackage{url}
\usepackage{comment}
\usepackage{mathtools}
\usepackage{epsf,pgf,graphicx}
\usepackage{color}
\usepackage{tikz,pgf}
\usepackage{diagbox}
\usepackage{makecell}
\usepackage{color}
\usepackage{braket}
\usepackage{mathtools}
\usepackage{float}
\usepackage{tikz}
\usetikzlibrary {positioning}
\usepackage{pgf}
\usepackage{qcircuit}
\usepackage{multirow}
\usepackage{multicol}
\usepackage{listings}
\usepackage{array}
\usepackage{makecell}
\newcommand{\hx}{{\mathbf{\hat{x}}}}
\newcommand{\tx}{{\mathbf{\tilde{x}}}}
\newcommand{\xv}{\mathbf{x}}

\newcommand{\ff}{\frac{1}{\sqrt{2}}}
\newcommand{\fs}{\frac{n}{4}}
\newcommand{\fl}{\lfloor \frac{n}{4} \rfloor}
\newcommand{\fr}{\lceil \frac{n}{4} \rceil}

\newcommand{\flt}{\lfloor \frac{3n}{4} \rfloor}
\newcommand{\frt}{\lceil \frac{3n}{4} \rceil}

\newcommand{\mb}[1]{\bm{#1}}
\newcommand{\ketm}[1]{\ket{\bm{#1}}_0}
\newcommand{\Q}{{Q_{\textrm{algo}}}}    
\newcommand{\QC}{{ QC_{\textrm{algo}}}}
\newcommand{\DP}{{D_{\op}^{(2)}}}
\newcommand{\op}{\oplus}
\newcommand{\cc}{{\sf C_0}}
\newcommand{\cd}{{\sf C_1}}
\newcommand{\cn}[2]{{\sf CNOT^{#1}_{#2}}}
\newcommand{\pg}[2]{{\sf P^{#1}_{#2}}}
\newcommand{\parg}[3]{{\sf Par^{#1}_{#2,#3}}}
\newcommand{\sg}[3]{{\sf S^{#1}_{#2,#3}}}
\DeclarePairedDelimiter{\abs}{\lvert}{\rvert}
\newcommand{\pdeg}{{\sf pdeg}}

\newcommand{\vph}{\vphantom{\frac{1}{2}}}

\DeclarePairedDelimiter\kett{\lvert}{\rangle}

\newcommand{\F}{\mathbb{F}_2}

\newtheorem{fact}{Fact}
\newtheorem{lemma}{Lemma}
\newtheorem{theorem}{Theorem}
\newtheorem{corollary}{Corollary}

\newtheorem{definition}{Definition}
\newtheorem{remark}{Remark}

\usepackage[widespace]{fourier}

\makeatletter
\providecommand{\nmid}{} 
\renewcommand\nmid{\mathrel{\mathpalette\thiel@nmid\relax}}
\newcommand{\thiel@nmid}[2]{%
  \ooalign{%
    \rotatebox[origin=c]{40}{$\m@th#1-\mkern-3.5mu$}\cr
    \hidewidth$\m@th#1|$\hidewidth\cr}%
}
\makeatother

\begin{document}
\title{Exact Quantum Query Algorithms Outperforming Parity -- Beyond The Symmetric Functions}

\author{
\IEEEauthorblockN{Chandra Sekhar Mukherjee,}
\IEEEauthorblockA{Indian Statistical Institute, Kolkata\\
Email: chandrasekhar.mukherjee07@gmail.com\\}
\and
\IEEEauthorblockN{Subhamoy Maitra,}
\IEEEauthorblockA{Indian Statistical Institute, Kolkata\\
Email: subho@isical.ac.in}}

\maketitle
\begin{abstract}
In Exact Quantum Query model, almost all of the Boolean functions for which non-trivial 
query algorithms exist are symmetric in nature. The most well known techniques in this 
domain exploit parity decision trees, in which the parity of two bits can be obtained by a single query. 
Thus, exact quantum query algorithms outperforming parity decision trees are rare. 
In this paper we first obtain optimal exact quantum query algorithms ($Q_{algo}(f)$) for a direct sum based class of 
$\Omega \left( 2^{\frac{\sqrt{n}}{2}} \right)$ non-symmetric functions.
We construct these algorithms by analyzing the algebraic normal form together with a novel untangling strategy.
Next we obtain the generalized parity decision tree complexity ($D_{\oplus}(f)$) analysing the Walsh Spectrum.
Finally, we show that query complexity of $Q_{algo}$ is $\lceil  \frac{3n}{4} \rceil$ whereas
$D_{\oplus}(f)$ varies between $n-1$ and $\lceil  \frac{3n}{4} \rceil+1$ for different classes,  
underlining linear separation between the two measures in many cases. To the best of our knowledge, this is the first family of 
algorithms beyond generalized parity (and thus parity) for a large class of non-symmetric functions. 
We also implement these techniques for a larger (doubly exponential in $\frac{n}{4}$) class of Maiorana-McFarland type
functions, but could only obtain partial results using similar algorithmic 
techniques. 
\end{abstract}


\section{Introduction}
\label{sec:1}
Query complexity considers the model of computation where a function is evaluated by an algorithm, 
either with certainty or with some probability. This is achieved exploiting a classical or a quantum 
computer such that the inputs to the function can only be accessed by making queries to an oracle. 
In this paradigm, the complexity of an algorithm is decided by ``the maximum number of queries it 
makes to an oracle to calculate the output for any input". Studying the classical-quantum 
separation in this model is of theoretical interest. 

In this paper our basic tool is Boolean functions expressed as $f: \mathbb{F}^n_2 \rightarrow \mathbb{F}_2$,
defined at all the points in $\mathbb{F}^n_2$. Since it is defined at all the points, sometimes we call these
Total Boolean functions. In this text, we will sometime simply use the word ``function" also. 
Boolean functions are the most widely studied class of functions in the query model. 
At the same time these are also of huge importance in the fields of cryptography and coding theory~\cite{pbook}. 
In fact, our motivation here is to explore the well studied functions in coding and cryptography to identify
the separation in terms of query complexity.

In the query complexity model, different aspects of Boolean functions are analyzed to design efficient quantum query algorithms. 
Examples are symmetry and Walsh Spectrum (Fourier spectrum) of the function in consideration. In this paper we concentrate on 
the $\mathbb{F}_2$ polynomial of a Boolean function. This is also called the algebraic normal form (ANF). Notably ANF is a 
very important property of Boolean functions from a cryptographic aspect~\cite{pbook}, in particular the algebraic degree. 
However ANFs are seldom used in complexity theoretic studies. Against this backdrop, a central theme of this paper is to use the 
ANF of the functions in the $\sf pdsp$ class to merge our positive and negative results. Our exact quantum query algorithm $\Q()$ 
is designed by analyzing the ANF of these functions. In this regard we have created a novel untangling protocol that gives us 
optimal complexity. On the other hand the functions in the $\sf pdsp$ class have high granularity. Granularity is a property of 
the Fourier spectrum of a function and higher value implies higher parity decision tree complexity. Now, the high granularity 
of the $\sf pdsp$ class is a direct consequence of its ANF structure. This result allows us to separate $\QC(f)$ and 
$D_{\oplus}(f)$, which, alongside the novel untangling protocol, is the main result of our paper. 

The query complexity model we discuss here is technically different from querying the universal quantum gate 
$U_f$, given a Boolean function $f$. Many important Quantum techniques, such as Deutsch-Jozsa~\cite{deutsch}, 
the Grover's search \cite{grover}, Simon's period finding algorithm~\cite{simon} and the hidden subgroup problem, 
underlying Shor's algorithm \cite{shor} are understood in a different setting. For detailed discussion related 
to quantum paradigm, the reader is referred to~\cite{nc}.

Given the model of complexity we discuss in this paper, the value of any variable can only be queried 
using an oracle. An oracle may be viewed as a black-box that can perform a particular computation.
In the classical model, an oracle can accept an input $i \ (1 \leq i \leq n)$ and output the value of 
the variable $x_i$. In the quantum model, the oracle is reversible and it can be represented as an unitary 
$O_x$ which works as:
$O_x\ket{i}\ket{\phi}=\ket{i}\ket{\phi \oplus x_i},~1 \leq i \leq n$.

The query complexity of a function can be defined as the maximum number of times this oracle 
needs to be used to obtain the output value of the function $f$ for any value of the variables 
$x_1, x_2, \ldots, x_n$. Naturally, the query complexity in quantum paradigm will be less than or equal 
to that in classical domain. It is easy to understand that the maximum number of queries for 
$n$ variables will be $n$. The question is how can we reduce the value over classical domain using 
the advantage of quantum computation. 

Query complexity can be defined for both deterministic and probabilistic classical 
computational settings, as well as in the bounded error quantum and exact quantum 
model. Out of these models, the Exact Quantum Query complexity model is perhaps
the least explored. Algorithms showing separations between the classical deterministic
and the exact quantum query model has been formulated for very few classes of functions.
In the exact quantum query model, a Boolean function $f$ needs to be evaluated correctly
for all possible inputs. The class of functions for which classical-quantum separation is known,
and more importantly, for which we have exact quantum algorithms which outperform 
the classical algorithms are far and few. 
Mostly the exact quantum query algorithms that exist use the same method 
of calculating of parity of two input bits in one query, as mentioned
in the work by Ambainis et. al.~\cite{amb2}.
\begin{quote}
``However, the techniques for designing exact quantum algorithms are rudimentary compared to the
bounded error setting. Other than the well known `XOR' trick — constructing a quantum algorithm from
a classical decision tree that is allowed to `query' the XOR of any two bits — there are few alternate
approaches."
\end{quote}
Even in this case, there is no generalized method for constructing parity decision trees
that exhibit deterministic-exact quantum advantage for a given class of functions. 
The most striking result in this area is the example of super-linear separation
between the deterministic classical and exact quantum models, shown in \cite{amb1}.
The work by Barnum et. al. \cite{sdp} is also equally important, defining a semidefinite programming 
formulation to find out the exact quantum query complexity of a given function and also discovering 
an algorithm to achieve it. Finding such separations remains the most interesting problem in this area. 
In terms of the gap between $D(f)$ and $Q_E(f)$, the separations can be distinguished into two different kinds. 

\begin{enumerate}
\item The first is identifying functions $f$
so that $Q_E(f) < \frac{D(f)}{2}$. As explained in~\cite{exact} this leads to
super-linear separation between $Q_E(f^k)$ and $D(f^k)$ where $f^k$ is 
obtained by recursively expanding the function $f$. 
\item The second is where $Q_E(f) \geq \frac{D(f)}{2}$. We do not have any known method of converting such a result
into super-linear separation. However, studying separation of this kind is still of considerable importance. 
This is because there are very few results in the exact quantum query model, and it is always of interest to find new 
approaches of evaluating functions in this model beyond the well known parity method, as highlighted in the quote
above~\cite{amb2}.
\end{enumerate}

Before proceeding further, let us first review the main results in this area in a chronological fashion to 
show where exactly our work is placed among the state of the art literature. 
\begin{itemize}
\item[]
{\bf 2012:} \cite{amb1} Superlinear separation between exact quantum query complexity ($Q_E(f)$) and deterministic query complexity 
($D(f)$) is obtained. This could be achieved by first obtaining a function $f$ with $Q_E(f) < \frac{D(f)}{2}$ 
and then recursively expanding it.

\item[]
{\bf 2013:} \cite{amb3} Exact quantum query complexities of the symmetric function classes $\sf Threshold^n_k$ and 
$\sf Exact^n_k$ have been obtained. For both the cases $Q_E(f) >\frac{D(f)}{2}$ and thus these results provided linear separation only. 

\item[] 
{\bf 2015:} \cite{amb4} Near quadratic separation between $Q_E(f)$ and $D(f)$ was obtained using the concept of pointer functions. 

\item[]
{\bf 2016:} \cite{amb2} Exact quantum query complexity of the symmetric function class $\sf Exact^n_{k,l}$ was obtained.
For all the functions we have $Q_E(f) > \frac{D(f)}{2}$ and thus the separation was linear.
\end{itemize}
As observed from the chronology above, discovering exact quantum query complexity, with only linear separation between the classical and 
quantum models, that is, ${Q_E(f)} > \frac{D(f)}{2}$, remained a relevant topic even after discovering the results 
related to near quadratic separation.

Against this backdrop, we study the exact quantum query model for Boolean functions by a combined analysis of $\F$ polynomial and 
Fourier spectrum to obtain linear separation between $Q_E(f)$ and $D(f)$ and show that our algorithms are more efficient than any 
parity decision tree method. In fact, the algorithms we design outperform generalized parity decision tree methods, 
where one can obtain the parity of any $i \leq n$ variables using a single query. Another interesting characteristics of 
the class of functions we obtain is that, the size of these classes are considerably larger compared to the symmetric function classes 
for which linear separations were previously obtained. The comparison of the existing results with our findings are 
summarized in Table~\ref{tab:1}.

\begin{table}[ht!]
\centering
\normalsize
\begin{tabular}{| c | c | c | c | c |}
\hline
Function & Ref. &\makecell{ Complexity of\\ Exact Quantum \\ Query Algorithm} & \makecell{Total \\ functions \\ covered for $n$} & \makecell{ Provably \\ Optimal?}  \\ \hline
$\sf Exact^n_k$  & \cite{amb3} & $\max\{k,n-k\}$  &\makecell{  $n+1$ \\ (one for\\  each value of $k$ )} & yes \\ \hline

$\sf Threshold^n_k$ & \cite{amb3} &  $\max\{k,n-k+1\}$  & \makecell{  $n+1$ \\ (one for \\ each value of $k$ )} & yes \\ \hline

$\sf Exact^n_{k,l}$ &  \cite{amb2}  &   $\max\{n-k,l\}+1$   &\makecell{  $n \choose 2$ \\ (one for \\  each $\{k,l\}$  pair)} & 
\makecell{For most \\ cases} \\ \hline

\makecell{ The class -- \\ $\sf pdsp$} & \makecell{our \\ work} & $\flt$ &$\Omega(\sqrt{2^{\sqrt{n}}})$ 
& yes \\ \hline

\makecell{ A subclass of \\  MM type \\ Bent functions} & \makecell{our \\ work} & $\lceil \frac{5n}{8} \rceil$ & $\Omega((2^{\fl}!)^22^{2^{\fl}})$ & No \\
\hline

\end{tabular}
\caption{Advantage achieved by Query Algorithms}
\label{tab:1}
\end{table}

Before proceeding further, we first define the following notations that we use in the paper.
\begin{definition}
\label{def:1} \
\begin{enumerate}
\item {$\bm{D(f)}$:} The Deterministic (classical) query complexity $D(f)$ of a Boolean function $f$ 
is the minimum number number of queries any classical algorithm must make 
to evaluate the function correctly for any possible input.

\item {$\bm{Q_E(f)}$:} The exact quantum query complexity $Q_E(f)$ of a Boolean function $f$ is the minimum 
number of queries any quantum algorithm must make to evaluate the function correctly for
any possible input. 

\item {$\bm{D_{\op}(f)}$ and $\bm{\DP(f)}$:} 
We define the generalized parity decision tree complexity $D_{\op}(f)$ of a function $f$ as 
the minimum number of queries any algorithm must make where the algorithm can 
obtain any parity $\op_{i \in S} x_i$ in a single query where $S$ is any subset of 
$[n]=\{1, \ldots, n\}$. 
If we restrict $\abs{S}=2$ then the algorithm is a parity decision tree, which is 
the most well known exact quantum query algorithm and we denote by $\DP(f)$ 
the minimum number of queries any parity decision tree needs to make to evaluate $f$. 
Consequently $D_{\op}(f) \leq \DP(f)$.

\item {$\bm{\Q(f)}$ and $\bm{\QC(f)}$:}
For any Boolean function $f$, we denote by $\Q(f)$ the exact quantum query algorithm 
designed in this paper to evaluate the function. 
The number of queries required by $\Q(f)$ is denoted with 
$\QC(f)$, which we denote as the query complexity of the
exact quantum query algorithm $\Q(f)$.
We call an algorithm $\Q(f)$ optimal if we have $\QC(f)=Q_E(f)$.
Here it should be clearly noted that $\Q(f)$ is an exact quantum query algorithm we design 
whereas $\QC(f)$ is a number, which is the query complexity of $\Q(f)$. 
\end{enumerate}
\end{definition}
We have the following relations between the aforementioned 
quantities.
\begin{fact}
For any Boolean function $f$ we have
\begin{multicols}{2}
\begin{enumerate}
\item $Q_E(f) \leq \DP(f) \leq D(f)$.
\columnbreak
\item $ D_{\op}(f) \leq \DP(f) \leq D(f)$.
\end{enumerate} 
\end{multicols}
\end{fact}
One should note here that unlike the situation of $Q_E(f) \leq \DP(f)$, there is no strict relationship between  $D_{\op}(f)$ and $Q_E(f)$. 
In fact for the simple parity function on $n$ variables, we have $Q_E(f)=\lceil \frac{n}{2} \rceil$ whereas it only takes a single 
query in the generalized parity decision tree model by definition, making $D_{\op}(f)=1$. Let us now lay out the organization and 
contribution of this paper.

\subsection{Organization \& Contribution}
In this paper we discuss the exact quantum query complexity of two classes of non-symmetric functions
which we analyze based on their $\F$ polynomial structure. We attempt to obtain $Q_E(f)$, $D_{\op}(f)$ and $D(f)$ of the functions 
and identify the situations with $\QC(f)=Q_E(f) < D_{\op}(f) \leq D(f)$. 
The motivation for this is threefold. 
\begin{enumerate}
\item To design non-trivial optimal exact quantum algorithms $\Q(f)$ for non-symmetric functions.
\item To identify situations where $\QC(f) < \DP(f)$. In fact, we discover classes for which $\QC(f) < D_{\op}(f)$. 
\item Design a class of algorithmic techniques, so that it can outperform parity decision tree method for a large number 
of functions for any $n$. 
\end{enumerate}
We summarize the list of our results in Table~\ref{tab:2}.
\begin{table}[H]
\centering
\normalsize
\begin{tabular}{| c | c | c | c | c | c | c |}
\hline
\makecell{ Functions} & Size & $Q_E(f)$ & $\QC(f)$ & $D_{\op}(f)$ &  
$\DP(f)$  & $D(f)$
\\ \hline 

\makecell{ The classes \\ $\sf{pdsp}(n,\frt,t+1)$, \\ $1 \leq t \leq \fl$} &
$\Omega \left( \sqrt{2^{\sqrt{n}}} \right)$ & 
$\flt$ & $\flt$ & $n-t$ & $ n-\lfloor \frac{t}{2} \rfloor$ & $n$
\\ \hline

\makecell{ A subclass of \\  MM type \\ Bent functions} & $\Omega((2^{\fl}!)^22^{2^{\fl}})$
& $\geq \lceil \frac{n}{2} \rceil$ & $\lceil \frac{5n}{8} \rceil$ & $\lceil \frac{n}{2} \rceil +1$
& $\leq \frt$ & $n$ \\ \hline
\end{tabular}
\caption{Advantage achieved by Query Algorithms}
\label{tab:2}
\end{table}

Let us now discuss these results in more details. We have worked with two Boolean function classes, 
the $\sf pdsp$ class and Maiorana-McFarland (MM) type bent functions.
In this direction, Section~\ref{sec:prelim0} explains these two important classes. 
As well, in Section~\ref{sec:prelim}, we describe the different unitary operations needed towards building the quantum algorithms. 

Section~\ref{sec:2} is on the $\sf pdsp$ class of functions. We build various algorithmic techniques needed 
towards proving Theorem~\ref{th:main}, which is the main contribution of this paper. First we explain the methodologies to obtain 
$D_{\op}(f)$ and $D(f)$ in Section~\ref{sec:2-1}. Then, in Section~\ref{sec:2-2}, we describe the oracle in the quantum query model 
and the registers on which quantum query algorithms are designed. Section~\ref{sec:2-3} presents the algorithmic techniques
to design the family of exact quantum algorithms and how these techniques are modified when dealing with different functions in question 
leading to the $\sf pdsp$ class. Finally, for the function $f$ in $\sf{pdsp}(n,\frt,t+1)$ class, we have 
obtained the following result in Theorem~\ref{th:main}.
\begin{enumerate}

\item 
We first observe $D(f)=n$ using the real polynomial representation of $f$.

\item 
Then we obtain that the generalized parity decision tree complexity  $D_{\op}(f)$ 
is $n-t$, which we derive using the concept of Granularity in the work~\cite{PAR} 
and results we obtain in Lemma~\ref{th:par}. 
In this direction it is easy to see that $\DP(f)$ is at most $n-\lfloor \frac{t}{2} \rfloor$.
Thus we have $n-t \leq  \DP(f) \leq n- \lfloor \frac{t}{2} \rfloor$ as we know $D_{\op}(f) \leq \DP(f)$. 

\item 
Next we observe that $Q_E(f) \geq \flt$ by reducing $f$ to the $\textrm{AND}_{\flt}$ function.
Finally we design an algorithm that reaches this query complexity using the untangling protocol described in Theorem~\ref{lemma:1}.
Thus, for $0 \leq t \leq \fl$ this is more efficient than any generalized parity decision tree method.

\end{enumerate}
We conclude with Corollary~\ref{cor:3} showing that there are $\Omega\left( \sqrt{2^{\sqrt{n}}} \right)$ such functions, 
for which this separation is achieved.

Next we discuss a subclass of the MM type bent functions. 
Section~\ref{sec:4} describes the results related to the MM bent functions borrowing certain ideas 
from Section~\ref{sec:2}. We first study the relation between $D(f), \DP(f), D_{\op}(f)$ and $Q_E(f)$ on a small number of variables, 
and then get into the generic discussion. In fact the study for $n = 4, 6$ had been our initial study that directed us to explore the
area in this manner. The effort starts with this subclass and we finally estimated that 
it contains a number of functions that is doubly exponential in $n$, as explained in Section~\ref{mm:num}. 
However, while we design an algorithm in Theorem~\ref{th:4} that evaluates any function in this class by making a total 
of $\lceil \frac{5n}{8} \rceil$ 
functions, our results are incomplete from two directions. 
\begin{itemize}
\item Firstly, we cannot obtain $Q_E(f)$ of the functions in these class. All we know is 
$\frac{n}{2} \leq Q_E(f) \leq \lceil \frac{5n}{8} \rceil$.
\item On the other hand, as described in~\cite{parity} we have $\DP(f) \leq \flt$, but we do not have any relevant lower bound 
on this measure. Furthermore, it is easy to see that $D_{\op}(f)$ is $\lceil \frac{n}{2} \rceil +1$ and thus even linear 
separation between $Q_E(f)$ and $D_{\op}(f)$ is not possible. Coming up with tight bounds for $Q_E(f)$ and $\DP(f)$ remains
the main open problem related to this MM class.
\end{itemize}
We conclude the paper in Section~\ref{sec:5}. 

Having discussed the function related results, we now talk about another important point related to our contribution. Let us
explain our novel algorithmic idea based on untangling of qubits that allows us to finally design algorithm whose complexity 
touches $Q_E(f)$ for the $\sf pdsp$ class.

\subsubsection*{Novel Algorithmic Techniques}
The exact quantum query algorithms for the symmetric functions such as Threshold or Exact as described in Table~\ref{tab:1} are designed 
by creating an equal superposition of all possible input states. Then some properties of the symmetric functions are exploited to obtain the 
desired result. In this paper we design a general class of algorithms that works in a completely different way. Informally, our algorithms 
are based on treating the function in $f$ in question as the direct sum of two functions $g$ and $h$, i.e., 
$f(\xv)=g(\hx) \op h(\tx)$ where the variables in $\xv$ are divided (partitioned) into two disjoint subsets $\hx$ and $\tx$. 
Then $(-1)^{g(\hx)}$ and $(-1)^{h(\tx)}$ are evaluated as relative phases at parallel using the superposition property of quantum 
computation which in its course entangles the system. Next we design a novel way of un-entangling the system described in 
Theorem~\ref{lemma:1}. We effectively negate the entangling due to $4$ variables using a single query in the functions we choose, 
which is driven by the fact that $\hx$ and $\tx$ do not share any variable. Repeated use of this technique is central to the 
separations we achieve.

\subsection{Boolean Function Classes}
\label{sec:prelim0}
\begin{definition}[The pdsp class]
\label{def:2}
First we define a perfect direct sum function to be a function on $n$ variables
such that all the variables $x_i, 1 \leq i \leq n$ appear only once in the function's
unique algebraic normal form ($\mathbb{F}_2$ polynomial).

Then a function $f$ is said to belong to the class ${\sf pdsp}(n,l,q)$
if the variable space $\xv=(x_1,x_2, \ldots, x_n)$ consisting of $n$ variables
can be divided into the two subspaces $\hx=(x_{r_1},x_{r_2},\ldots, x_{r_l})$ and 
$\tx=(x_{r_{l+1}},x_{r_{l+2}}, \ldots ,x_{r_n})$ containing $l$ and $n-l$ variables 
respectively so that
\begin{enumerate}
\item $f(\xv)= f_1(\hx)f_2(\tx)$. 
\item $f_1$ is a perfect direct sum function defined on the $l$ variables 
$\hx$, which consists of $q$ monomials such that each monomial consists 
of at least $q$ variables. 
\item $f_2$ is the product function of the $n-l$ variables in $\tx$. 
That is, $f_2(\tx)=\prod\limits_{i=l+1}^n x_{r_i}$.
If $l=n$ then $f_2$ function is not defined. 
\end{enumerate}
\end{definition}
We observe that these functions have high granularity as well as their $\F$ polynomial structure agrees with our exact quantum query algorithm, and 
this leads us to Theorem~\ref{th:main} which provides the provable separations noted in Table~\ref{tab:2}. 
To see that $f$ (as in Theorem~\ref{th:main}) is indeed a function in ${\sf pdsp}(n,\frt,t+1)$,
it suffices if we define 
$\hx= \left(x_1, \ldots ,x_{\frac{n}{2}}, x_{\flt+1}, \ldots x_n \right)$ and 
$\tx=\left( x_{\frac{n}{2}+1}, \ldots x_{\flt} \right)$. Then We have 
$f_1(\hx)=\prod_{i=1}^{\frac{n}{2}} x_i \bigoplus g(\xv')$
and $f_2(\tx)=\prod_{j=\frac{n}{2}+1}^{\flt} x_j$, and thus $f_1$ 
is a perfect direct sum function containing $t+1$ monomials so that the degree of each monomial is at least $t+1$. 
Considering $g(\xv')=\prod_{k=\flt+1}^n x_k$, we get the function 
 $f(\xv)= \prod_{i=1}^{\flt} x_i \op \prod_{j=\frac{n}{2}+1}^n x_j$ for which we obtain the desirable separations.

\begin{definition}[MM type functions]
For any two positive integers $n_1$ and $n_2$ with $n_1+n_2=n$,
An MM Boolean function on $\mathbb{F}_2^n$ is defined as
$$f(\hx,\tx)=  \phi(\hx) \cdot \tx \op g(\hx)$$
where the subspaces are
$\hat{x} \in \mathbb{F}_2^{n_1}$, $\tilde{x} \in \mathbb{F}_2^{n_2}$,
$g$ is any Boolean function defined on $F_2^{n_1}$
and $\phi$ is a map of the form $\phi: \mathbb{F}_2^{n_1} \rightarrow \mathbb{F}_2^{n_2}$.
\end{definition}
Here $a \cdot b$ is the dot product of two $n_2$ dimensional vectors, defined as 
$a \cdot b=a_1b_1 \op a_2b_2 \ldots a_{n_2}b_{n_2}$.
If we set $n_1=n_2$ and restrict $\phi$ to be a bijective map, then all resultant 
Boolean functions are bent. These are the functions with highest possible nonlinearity for any even $n$ \cite{bent1}. 
 
Now consider the quadratic bent functions on $4$ variables, $f_{id}^4(x) = x_1x_3 \op x_2x_4$ 
and on $6$ variables, $f_{id}^6(x)=x_1x_4 \op x_2x_5 \op x_3x_6$. 
We observed that for any MM type bent function there is a parity decision tree of query complexity $\frt$.  
For the case of $f_{id}^4$, this strategy proved to be optimal, as verified through semi-definite programming 
formulation of \cite{sdp}. 
However for $f_{id}^6$ we note that $Q_E(f_{id}^6)=4$, whereas the parity decision technique has a query 
complexity of $5$. This motivated us to study how to design an exact quantum query algorithm $\Q$ based on the $\F$ polynomial 
structure of these functions, so that it matches the $Q_E$ value for $f_{id}^6$, and then generalized it to it a novel untangling based algorithmic technique that we have discussed in
Theorem~\ref{lemma:1} which has given us the efficient quantum query algorithms for the said $\sf pdsp$ classes. 

\subsection{Some Unitary Matrices}
\label{sec:prelim}
\subsubsection*{$\sg{n}{i}{j}$}
This operation is also defined on an $n+1$ dimensional register. 
It only performs the transformation 
$~~ \ket{\mb{0}} \xrightarrow{\sg{n}{i}{j}} \ff ( \ket{\mb{i}} +\ket{\mb{j}} )$.

The corresponding matrix can fairly simply be defined as 
$\sg{n}{i}{j}=\big( \parg{n}{i}{j} \big)^*$.
Below are examples of the matrices corresponding to operations of type $\pg{n}{i}$, $\parg{n}{i}{j}$ and
$\sg{n}{i}{j}$.
\begin{multicols}{3}

$$\pg{4}{3}=
\begin{bmatrix}
1 &  0 & 0 & 0 & 0 \\
0 &  0 & 0 & 1 & 0 \\
0 &  0 & 1 & 0 & 0 \\
0 &  1 & 0 & 0 & 0 \\
0 &  0 & 0 & 0 & 1 
\end{bmatrix}
$$
\vfill\null
\columnbreak
$$\parg{4}{1}{3}=
\begin{bmatrix}
0   & \ff & 0 & \ff  & 0 \\
0   & \ff & 0 & -\ff & 0 \\
1   & 0   & 0 & 0    & 0 \\
0   & 0   & 1 & 0    & 0 \\
0   & 0   & 0 & 0    & 1 
\end{bmatrix}
$$

\vfill\null
\columnbreak
$$\sg{4}{0}{1}=
\begin{bmatrix}
\ff & \ff & 0 & 0 & 0 \\
\ff &-\ff & 0 & 0 & 0 \\
0   &   0 & 1 & 0 & 0 \\
0   &   0 & 0 & 1 & 0 \\
0   &   0 & 0 & 0 & 1 
\end{bmatrix}
$$
\vfill\null

\end{multicols}

To avoid confusion, it should be noted that in a Quantum
Computer built with qubits, a unitary operator working
on $z$ qubits is has a dimension of $2^z$.
In the following remark we state how a general $n$ dimensional
register can be implemented in such a setup.

\begin{remark}
To implement an unitary operator  $U$ on an $n+1$ dimensional register, it would require $\lceil \log(n+1) \rceil$ qubits and
thus the corresponding unitary matrix 
$U'$ being applied on these qubits 
would actually be $2^{\lceil \log(n+1) \rceil}$ dimensional.
The matrix  $U'$ can be formed by adding
$2^{\lceil \log(n+1) \rceil}-(n+1)$ rows and columns
to $U$, such that entries in the
rows and columns corresponding to the basis states $\ket{i}, i \in\{ n+1, 
\ldots ,\lceil \log(n+1) \rceil-1 \}$
would simply be $U'(i,i)=1$.
\end{remark}
For example, given the matrix $\pg{2}{2}$,
the matrix $\pg{2}{2}'$ that would be implemented in a $2^{\lceil \log(3+1) \rceil}=4$ dimensional system built of $2$ qubits is as follows:

\begin{multicols}{3}

$$\pg{2}{2}=
\begin{bmatrix}
1 &  0 & 0 \\
0 &  0 & 1 \\
0 &  1 & 0  
\end{bmatrix}
$$ 

\columnbreak

$$\pg{2}{2}'=
\begin{bmatrix}
1 &  0 & 0 & 0 \\
0 &  0 & 1 & 0 \\
0 &  1 & 0 & 0 \\
0 &  0 & 0 & 1 
\end{bmatrix}
$$

\columnbreak

\end{multicols}

\subsubsection*{$\cc U$ and $\cd U$}

The algorithm we design uses unitary matrices that are controlled
on the state of a work qubit $w_1$.
At each step we apply a set of unitaries controlled on $w_1=\ket{0}$
and another set controlled on $w_1=\ket{1}$.
Given any unitary $U$, we denote by $\cc U$ the operation that is 
$U$ controlled on $w_1=\ket{0}$.
We use the notation $\cd U$ to denote the operation that is 
$U$ controlled on $w_1=\ket{1}$.
It is easy to see that if $U$ is a unitary operation, then so is
$\cc U$ and $\cd U$.

\subsubsection*{$\cn{a}{b}$}

This operation is the Controlled-NOT operation from register $a$ to register $b$,
with register $a$ as control and register $b$ as target.
Here either one of the registers is the query
register, or else both the registers are work qubits.
Let us suppose register $a$ is the query register,
then the transformation will be denoted by
\begin{itemize}
\item $ \ket{\mb{1}}_a\ket{0}_b \rightarrow \ket{\mb{1}}_a\ket{1}_b$
\item $ \ket{\mb{1}}_a\ket{1}_b \rightarrow \ket{\mb{1}}_a\ket{0}_b$
\end{itemize}

If both the registers are qubits, then it works as the conventional C-NOT operation.
If $a$ and $b$ are both qubits, then it is a $4$ dimensional
unitary operation, otherwise it is a $2(n+1)$ dimensional
operation.

\subsubsection*{$swap(a,b)$}

This operation is simply defined as 
$swap(a,b)=\cn{a}{b}~\cn{b}{a}~\cn{a}{b}$, and swaps the value of two registers $a$ and $b$, 
with dimensions $d_1$ and $d_2$, so that it is defined on the 
computational basis states $\ket{i} \otimes \ket{j}: i,j \leq \min (d_1,d_2)$.

\section{Results on the $\sf pdsp$ class of functions} 
\label{sec:2}
In this section we build till Theorem~\ref{th:main},
by obtaining different query complexity measures, describing the properties of 
Boolean functions and the algorithmic techniques designed
for the $\sf pdsp$ class of functions.

First we obtain the deterministic query complexity $D(f)$ and 
generalized parity query complexity $D_{\op}(f)$ of this class of functions.

\subsection{$D_{\op}(f)=n-t$ and $D(f)=n$ for the $\sf pdsp$ class}
\label{sec:2-1}
We first obtain the deterministic query complexity of the functions 
defined in Definition~\ref{def:2} by analyzing the polynomial degree of the function.
Let $\pdeg(f)$ be the degree of the unique real multi-linear
polynomial $p: \F^n \rightarrow \mathbb{R}$ such that $f(\xv)=p(\xv)~ \forall \xv \in \F^n$.
From~\cite{polypoly}, we know $D(f) \geq \pdeg(f)$. Then we have the following result.
\begin{theorem}
\label{th:df}
For any function $f \in {\sf pdsp}(n,l,q)$ we have $\pdeg(f)=n$.  
\end{theorem}
Now let us discuss the generalized parity decision tree complexity $D_{\op}(f)$
for these functions. Naturally this also works as lower bound on the parity decision
tree complexity $\DP(f)$. We first obtain lower bounds on $D_{\op}(f)$
by analyzing the granularity of the functions as described in~\cite{PAR}
and then show that these bounds are in fact tight for $\sf pdsp$ class.
Let us first define the notion of granularity and the result.
\begin{definition}\cite{PAR}
\label{def:gran}
\begin{enumerate}
\item For a set $S \subseteq [n]$ the Fourier Character 
at $S$ is defined as 
$\chi_S(\xv)=\displaystyle \prod_{i \in S} (-1)^{x_i}$.

\item Given a function $f$ on $n$ variables its Fourier coefficient 
on the $S \subseteq [n]$ is denoted as 
$\hat{f}(S)= \frac{ \displaystyle \sum_{\xv \in \{0,1\}^n} \left( (-1)^{f(\xv)}\chi_S(\xv) \right) }{2^n}$.~~

\item The granularity of a rational number $t$ is defined as $gran(f)=k$ 
where $k$ is the smallest power of two such that $t \times k$ is an integer. 
\end{enumerate}
\end{definition}
Having noted the notations used in \cite{PAR}, we now refer to the following result.
\begin{theorem}
\cite{PAR}
\label{th:par2}
The general parity complexity of a Boolean function $f$ on $n$ variables is lower
bounded by $D_{\op}(f) \geq gran_m(f)+1$ 
where $gran_m(f)= \displaystyle \max_{S \subseteq [n]} gran(\hat{f}(S))$.
\end{theorem}
We  observe that the bound due to Theorem~\ref{th:par2} 
is tight for any perfect direct sum function, and then also 
obtain that it is tight for the $\sf pdsp$ class of functions 
defined in Definition~\ref{def:2}.

Interestingly we also identify that the result for this class is even more specific,
in the sense that the granularity is maximum for $\hat{f}(\phi)$. That is, 
the lower bound on $D_{\op}(f)$ is solely based on the number of 
ones in the truth tables of the functions. 
In this regard we have the following result.
\begin{lemma}
\label{th:par}
Let $f$ be a function defined on $n$ variables such that
$f \in {\sf pdsp}(n,l,q)$.
Then we have $D_{\op}(f)=gran_m(f)+1= gran(\hat{f}(\{ \phi \})+1 =n-q+1$.
\end{lemma}
\begin{proof}
We prove this by first showing $gran(\hat{f}(\{ \phi \}) = n-q$ which implies $D_{\op}(f) \geq n-q+1$
and then describe a simple general parity decision tree with complexity $n-k+1$.
For simplicity let us assume $r_i=i$.

\paragraph*{\bf $gran_m(f) \geq n-k$}
The ANF of $f_1$ can be represented as a partition of $\hx=\{x_{r_1},x_{r_2},\ldots x_{r_l}\}$ 
into $q$ disjoint sets. We denote these sets as $m_i, 1 \leq i \leq q$,
where $m_i$ consists of $q_i$ variables.
Then $f$ can be written as 
$$f(\xv)=\left(\bigoplus_{i=1}^q \left( \prod\limits_{x_j \in m_i} x_j \right) \right) \prod\limits_{p=l+1}^n x_p.$$
\
We know that $\hat{f}(\{ \phi \})= \sum\limits_{\bm{a} \in \{0,1\}^n}(-1)^{f(\bm{a})}= 2^n- 2wt(f)$ where $wt(f)$ is the number of input points 
for which the function outputs $1$.

The output of $f_1$ for some input is $1$ if some odd number of these 
$q$ monomials are evaluated to $1$ and $x_{r_p}=1,~ l+1 \leq p \leq n$.
Let us denote by $x^j$ all inputs from $\{0,1\}^l$ for which $j$ of the said monomials evaluate to $1$.
If $j$ is odd then for each input in $x_j$ $f_1$ evaluates to $1$. 
For each such $a \in x^j$, there is only input  $a' \in \{0,1\}^n$ 
for which $f$ evaluates to $1$, where $a'=a||1_{n-l}$.
We can represent the number of ones in the truth table of $f$ as
$\sum\limits_{i=0}^{\lfloor \frac{q-1}{2} \rfloor} \abs{x^{2i+1}}.$
Then we have $\abs{x^1}= \sum\limits_{i=1}^q \left(  \prod\limits_{j \neq i} (2^{q_j}-1) \right)$ as
it consists of inputs for which exactly one monomial has all variables set to one, and because 
of the monomial disjoint nature of the function there is no repetition in the counting.
We can express $x^1$ as
\begin{itemize}

\item[] $\abs{x^1}=\alpha_1 2^q +q$ such that $\alpha_1$ is an integer, if $q$ is odd.

\item[] $\abs{x^1}=\alpha_1 2^q -q$ such that $\alpha_1$ is an integer, if $q$ is even.

\end{itemize}
This is because in expansion of $\abs{x^1}$ in each product term 
we have a $(-1)^{q-1}$ ( $-1$ if $k$ is even $+1$ otherwise ) 
and all other terms are of the form $\pm 2^{q_{i_1} + q_{i_2} \ldots + q_{i_j}}$.
since $q_i \geq q~\forall i$, all these terms are integer multiple of $2^q$,
and thus their sum is also an integer multiple of $2^q$, or zero.
Now since each product term has a $(-1)^{q-1}$ and
there are $q$ terms in the expansion can be written as some $\alpha_1 2^q + (-1)^{q-1} q$.
Similarly, $x^i$ can be expressed as $\alpha_i2^q + (-1)^{q-1} {q \choose i}$
and therefore the support set of $f$ is of the size
\
$\displaystyle \sum_{i=0}^{\lfloor \frac{q-1}{2} \rfloor} \left( \alpha_{2i+1}2^q +(-1)^{q-1}{q \choose {2i+1}} \right)
=\alpha2^q + (-1)^{q-1} 2^{q-1}$.
\
Therefore the Fourier coefficient of the function at $S=\{\phi \}$ is 
\
\begin{align*}
\widehat{f}(\{\phi \})=\frac{2^n-2\left(\alpha2^q + (-1)^{q-1} 2^{q-1}  \right)}{2^n}
= \frac{2^n - \alpha2^{q+1} + (-1)^q 2^q}{2^n}.
\end{align*}

Thus granularity of the Fourier coefficient at this point is $gran(\widehat{f}(\{\phi \}))=n-q$
and therefore $gran_m(f) \geq n-q+1$.

\paragraph*{$D_{\op} \leq n-k+1$}
We now show a simple general parity tree of with $n-q+1$ queries 
that evaluates $f$, showing $D_{\op}(f) \leq n-q+1$.
Given an input $\mathbf{a}= \{a_1, a_2 ,\ldots , a_n \}$
It first queries all but one variable from each monomial of $f_1$. 
This takes $l-q$ queries.
For the monomial $m_i$ the product of these variables evaluate to
$\tilde{m_i}=\prod_{j=1}^{q_i-1} a_{i_j}$.
Then only if $\tilde{m_i}=1$ the output of $f_1$ depends on $x_{i_{q_i}}$. 
Therefore the final query to evaluate $f_1$ is the linear 
function $\bigoplus_{i=1}^q \tilde{m_i} x_{i_{q_i}}$ as the value of 
$\tilde{m_i}$ are already calculated. Thus evaluating $f_1$ needs $l-q+1$ queries.
Now we can simply evaluate $f_2$
which is defined on $n-l$ variables by querying each of the variables 
individually which enables us as to output the function $f(\xv)=f_1(\hx)f_2(\tx)$.
Therefore this method requires a total of $l-q+1+n-l=n-q+1$ query, which shows $D_{\op}(f) \leq n-q+1$. 

Since $D_{\op}(f) \geq gran_m(f)$ and we have $gran_m(f) \geq n-q+1$ and $D_{\op}(f) \leq n-k+1$
this implies $D_{\op}(f)=gran_m(f)=n-q+1$.
\end{proof}

Having determined $D(f)$, $D_{\op}(f)$ and $\DP(f)$, we now describe 
the family of exact quantum query algorithms $\Q$ that we design,
and their query complexity $\QC(f)$.
Let us now proceed with the functionality of the oracle and the different 
registers that are used by $\Q$.

\subsection{Quantum Query Algorithms} 
\label{sec:2-2}
The set-up for a Quantum Query algorithm in relation to Boolean functions is as follows. 
Given a Boolean function on $n$ influencing variables, 
a Quantum Query Algorithm for evaluating the function is defined 
on the Hilbert space $H=H^a \otimes H^q \otimes H^w$.
\begin{itemize}
\item Here $H^a$ represents an $n$ qubit register that contains the input to the function.
The inputs stored in the input register can only be accessed using the oracle $O_x$,
which operates on $H^a \otimes H^q$.

\item The Hilbert space $H^q$ is $n+1$ dimensional, 
and can be indexed with the basis states $\ket{\pmb{0}}$ $\ket{\pmb{1}}, \ket{\pmb{2}}, \ldots \ket{\pmb{n}}$.
This space is used to make queries to the oracle and we call this the query register $Q_n$.

\item The Hilbert space $H^w$ is used as the working memory  and has no restrictions.
We define $H^w$ to be formed of some $w$ qubits,
where the basis states of a qubit is $\ket{0}$ and $\ket{1}$
\footnote{
Therefore a Quantum Query Algorithm corresponding to a function $f$ with $n$ influencing variables
is essentially a circuit defined on the Hilbert space $H$ of  $n+ \lceil \log (n+1) \rceil +w $ qubits
with the restriction that the $n$ qubits corresponding to the input register
can only be accessed through an oracle.}.
\end{itemize}
The oracle $O_x$ works on the space $H^a \otimes H^q$ in the following way.
\begin{itemize}
\item $1 \leq i \leq n : O_x \ket{\xv}\ket{\pmb{i}}\ket{w}=(-1)^{x_i} \ket{\xv}\ket{\pmb{i}}\ket{w}$.

\item $i=0 : O_x \ket{\xv}\ket{\pmb{0}}\ket{w}=\ket{x}\ket{\pmb{0}}\ket{w}$.
\end{itemize}

Since the input register remains unchanged throughout the algorithm, 
we describe our algorithm on $H^q \otimes H^w$,
and describe the working of the oracle as $O_x\ket{\pmb{i}}=(-1)^{x_i}\ket{\pmb{i}}, ~ 1 \leq i \leq n$
and $O_x\ket{\pmb{0}}=\ket{\pmb{0}}$.

An algorithm that uses the oracle $k$ times can be expressed as a series of unitaries 
$U_0, U_1, \ldots U_k$ applied on $H^q \otimes H^w$ with an oracle access between 
each $U_i$ and $U_{i+1},~ 0 \leq i \leq k-1$.
The algorithm starts with the state $\ket{\psi}=\ket{\pmb{0}}\ket{0}\ldots \ket{0}$
and finally reaches the state $U_k O_x U_{k-1} \ldots U_1 O_x U_0 \ket{\psi}$,
on which some measurement is performed and the output is decided 
depending on some predefined condition on the result. 

An exact quantum query algorithm is one which evaluates a function correctly for
any input. The Exact Quantum Query complexity ($Q_E$) of a function is the least possible
number of queries an exact quantum query algorithm needs to make at most to
evaluate the function in any point. 

\subsubsection{The workspace of $\Q$}
The workspace of the algorithms that we design consists of 
$Q_n$ and $l+1$ qubits for some $l$. 
In this paper we denote the basis states of the query register with 
$\ket{\pmb{0}}_0,\ket{\pmb{1}}_0,\ldots, \ket{\pmb{n}}_0$. 
The qubits are denoted by $w_1$ through $w_{l+1}$. 
We denote the computational basis states of the 
$i$-th work qubit by $\ket{0}_i$ and $\ket{1}_i$.
Thus we describe this system with the basis states
$$\ket{\pmb{a_0}}_0 \bigotimes_{i=1}^t\ket{a_i}_i,
~a_0 \in \{0,1,\ldots, n\},
~ a_i \in \{0,1\}~ \forall i \in \{1,\ldots ,l+1 \}.$$

\subsection{Constructing $\Q$ leading to the $\sf pdsp$ class} 
\label{sec:2-3}

\begin{remark}
\label{rem:n}
From here on we assume $n \equiv 2 \bmod 4$. This is to simply reduce the tediousness 
of the proof. For other cases the algorithms and the bounds develop in an almost identical manner, conforming to the same generalized query complexity value. 
One can refer to the extended version of this work~\cite{self} (uploaded as an earlier Arxiv version)to view the simple modifications needed to incorporate the other cases.   
\end{remark}

The general flow of the algorithms are as follows.
\begin{itemize}

\item The function is expressed as $f(\xv)=g(\hx) \op h(\tx)$
where $\hx$ and $\tx$ are two subspaces that form a disjoint partition of $\xv$.

\item We then start with the state $\ket{\pmb{0}}_0 \otimes _{i=1}^{s+1} \ket{0}_i$
where $s$ is dependent on the structure of the function.

\item We apply a Hadamard gate on the first work qubit $w_1$ to obtain the state
\begin{equation*}
\ket{\psi_0}=
\ff \left( \ket{\pmb{0}}_0\ket{0}_1 \otimes _{i=1}^s \ket{0}_{i+1}+ 
\ket{\pmb{0}}_0 \ket{1}_1 \otimes _{i=1}^s \ket{0}_{i+1} \right).
\end{equation*}

\item Then we apply certain transformations to obtain a state of the form
\begin{equation}
\label{eq:01}
\ket{\psi_f}=
\ff \left( (-1)^{g(\hx)} \ket{\pmb{0}}_0\ket{0}_1 \otimes _{i=1}^s \ket{x_{r(i)}}_{i+1}+ 
(-1)^{h(\tx)} \ket{\pmb{0}}_0 \ket{1}_1 \otimes _{i=1}^s \ket{x_{r(l+i)}}_{i+1} \right)
\end{equation}
using some $k \leq s$ queries.
Here $r(i), 1 \leq i \leq 2s$ are elements of an injective map $r: [n] \rightarrow [n]$.
\end{itemize}
At this stage if we had $x_{r(l+i)}=x_{r(i)}=m_i \text{ (say)  }  \forall i$ we could write the 
state as 
$
\ket{\pmb{0}}_0 \ff 
\left( (-1)^{g(\hx)} \ket{0}_1 + (-1)^{h(\tx)} \ket{1}_1 \right)
 \bigotimes _{i=1}^s \ket{m_i}_{i+1}
$
and we could simply apply a Hadamard gate on the $w_1$ to obtain 
$
\ket{\pmb{0}}_0 \ff 
\kett*{{g(\hx)} \op {h(\tx)}}  \bigotimes _{i=1}^s \ket{m_i}_{i+1}
$
(ignoring global phase)
and measuring $w_1$ would suffice.
However we do not have any way of ensuring $x_{r(i)}=x_{r(l+i)}$ which leaves the state $\ket{\psi_f}$ 
in an entangled form. Here we design a new un-entangling protocol that finally gives us the separations.

\subsubsection{The un-entangling protocol}

Our algorithm is currently in the state
$$
\ket{\beta_0}=
\ff \left( (-1)^{g(\hx)} \ket{\pmb{0}}_0\ket{0}_1 \otimes _{i=1}^s \ket{x_{r(i)}}_{i+1}+ 
(-1)^{h(\tx)} \ket{\pmb{0}}_0 \ket{1}_1 \otimes _{i=1}^s \ket{x_{r(l+i)}}_{i+1} \right).
$$
Here it is important to note that $x_{r(i)} \in hx$ and $x_{r(l+i)} \in \tx$, that is the possible 
values of the qubits in the superposition states are also decided by our partition of $\xv$ into 
$\hx$ and $\tx$.
If the system was in a product state at this stage, we could have simply obtained the parity 
of the phases $(-1)^{g(\hx)}$ and $(-1)^{h(\tx)}$, which would have given us the desired outcome. 
However, the system is entangled as the value of $x_{r(i)}$ may differ depending on the input 
on which we have no control. At this stage we design a technique of untangling two qubits 
deterministically using a single query. This can be summarized as follows. 

\begin{theorem}
\label{lemma:1}
Let a quantum query algorithm be in the state 
$$\ket{\gamma} =\ff \big(\ket{\mb{x_a}}_0\ket{0}_1\ket{x_b}_2 \ket{W_1} + \ket{\mb{x_c}}_0\ket{1}_1\ket{x_d}_2\ket{W_2} \big)$$
Here $x_a$, $x_b$, $x_c$ and $x_d$ are inputs to a function corresponding to an oracle.
Then this state can be transformed to
$$\ket{\gamma'}=(-1)^{x_b} \ff \big(\ket{\mb{x_b}}_0\ket{0}_1\ket{x_d}_2\ket{W_1} + \ket{\mb{x_b}}_0\ket{1}_1\ket{x_d}_2\ket{W_2} \big)$$
using a single query to the oracle.
Here $\ket{W_1}$ and $\ket{W_2}$ represent any two arbitrary $m$-qubit states.
\end{theorem}
\begin{proof}
We again define a protocol, $\sf untangle$ which enables the defined 
transformation by making a single query to the oracle.

We first define the unitaries $U_1$ and $U_2$ that act on $Q_n$.
The structure of $\sf untangle$ is as follows.
First $\cc U_1$ and $\cd U_2$ are applied, 
followed by the oracle $O_x$ and then
$\cc \parg{n}{a}{d}$ and $\cd \parg{n}{b}{c}$.
That is, we define
$${\sf untangle}= \left( \cc \parg{n}{a}{d}~\cd \parg{n}{b}{c}~O_x~\cc U_1~\cd U_2 \right),$$
and show that $\ket{\gamma} \xrightarrow{\sf untangle} \ket{\gamma'}$.

We denote $\ket{\mb{x_a}}_0\ket{0}_1\ket{x_b}_2 \ket{W_1}= \ket{\gamma_1}$ and 
$\ket{\mb{x_c}}_0\ket{1}_1\ket{x_d}_2 \ket{W_2}= \ket{\gamma_2}$.Then 
$\ket{\gamma}=\ff\left( \ket{\gamma_1} +\ket{\gamma_2} \right)$
Let us now observe the evolution of the two states $\ket{\mb{x_a}}_0\ket{0}_1\ket{x_b}_2$ and
$\ket{\mb{x_c}}_0\ket{1}_1\ket{x_d}_2$ individually, 
depending on the state of the $w_1$.

We start with the case when $w_1=\ket{0}$.
$U_1$ can be looked as the composition of two operations $U_{10}$ and $U_{11}$.
$U_{10}$ and $U_{20}$ acts on the register $Q_n$
depending on if $w_2=\ket{0}$ or $\ket{1}$, i.e. $x_b=0$ or $x_b=1$
respectively. 
That is $U_{10}$ and $U_{20}$ are operators acting on $H^q \otimes H_2$.
Therefore at any point, only one of the unitaries actually perform 
their transformations, depending on the value of $x_b$. These transformations are
defined as follows.

\begin{multicols}{2}
 
\paragraph{$U_{10}$}
\begin{enumerate}

\item $\ket{\mb{0}}_0 \rightarrow \ff (\ket{\mb{a}}_0+\ket{\mb{d}}_0)$

\item $\ket{\mb{1}}_0 \rightarrow \ff (-\ket{\mb{a}}_0+\ket{\mb{d}}_0)$

\end{enumerate}

\columnbreak

\paragraph{\bf $U_{11}$}
\begin{enumerate}

\item $\ket{\mb{0}}_0 \rightarrow \ff (-\ket{\mb{a}}_0+\ket{\mb{d}}_0)$

\item $\ket{\mb{1}}_0 \rightarrow \ff (\ket{\mb{a}}_0+\ket{\mb{d}}_0)$

\end{enumerate}

\end{multicols}

That is, 
\begin{itemize}

\item $\ket{\mb{x_a}}_0 \xrightarrow{U_{10}} \ff ( (-1)^{x_a} \ket{\mb{a}}_0+\ket{\mb{d}}_0)$

\item $\ket{\mb{x_a}}_0 \xrightarrow{U_{11}} \ff ((-1)^{x_a+1} \ket{\mb{a}}_0+\ket{\mb{d}}_0)$

\end{itemize}

The oracle is then applied on $\cc U_{10} \cc U_{11} \ket{\gamma_1}$,
followed by the Unitary Operation $\cc \parg{n}{a}{d}$.
We now observe the state $\big( \cc \parg{n}{a}{d} O_x \cc U_{10} \cc U_{11} \big) \ket{\gamma_1}$,
depending on the value of $x_2$ and compare the resultant 
state with $$(-1)^{x_b} \big(\ket{\mb{x_b \op x_d}}_0)\ket{0}_1\ket{x_b}_2.$$
We tabulate the comparisons for both $x_b=0$ and $x_b=1$ in 
Table~\ref{tab:3}.
The transformations $\cc U_10, \cc U_{11}, O_x$ and $\cc \parg{n}{a}{d}$ only act on 
the query register, depending on the values of the qubits $w_1$ and $w_2$,
which remain unaltered throughout. Therefore we only show the evolution
of the query register.

\begin{table}[H]

\begin{center}
\begin{tabular}{ |c|c|c|c|c| } 
\multicolumn{1}{l}{$\bf x_b=0:$}\\
 \hline
 $x_a$  & $\cc U_{10}\ket{\gamma_1}$  & $O_x \cc U_{10}\ket{\gamma_1}$ & $\cc \parg{n}{a}{d} O_x \cc U_{10} \ket{\gamma_1}$ & $\ket{\mb \beta}$ \\ \hline
 $0$ & \makecell{$\ff\ket{\mb{a}}_0$ \\+ $ \ff \ket{\mb{d}}_0$}  
 & \makecell{$\ff (-1)^{x_a}\ket{\mb{a}}_0$   \\ + $\ff (-1)^{x_d}\ket{\mb{d}}_0$}   
 & \makecell{$(-1)^{x_a}\ket{\mb{x_a \op x_d}}_0$\\= $\ket{\mb{x_d}}_0$}
 & $\ket{\mb{x_d}}_0$ \\ \hline
 $1$ & \makecell{$-\ff \ketm{a}$\\+$\ff \ketm{d}$} 
 & \makecell{$\ff (-1)^{x_a+1}\ketm{a}$ \\+ $\ff (-1)^{x_d}\ketm{d}$} 
 & \makecell{$(-1)^{x_a+1}\ketm{x_a \op x_d \op 1}$\\= $\ketm{x_d}$}
 & $\ketm{x_d}$ \\ \hline
\end{tabular}
\end{center}

\begin{center}
\begin{tabular}{ |c|c|c|c|c| } 
\multicolumn{1}{l}{$\bf x_b=1:$}\\
 \hline
 $ x_a $  & $\cc U_{11}\ket{\gamma_1}$  & $O_x \cc U_{11}\ket{\gamma_1}$ & $\cc \parg{n}{a}{d} O_x \cc U_{11} \ket{\gamma_1}$ & $\ket{\mb \beta}$ \\ \hline
 $0$ & \makecell{$-\ff\ketm{a}$ \\+ $ \ff \ketm{d}$}  
 & \makecell{$\ff (-1)^{x_a+1}\ketm{a}$   \\ + $\ff (-1)^{x_d}\ketm{d}$}   
 & \makecell{$(-1)^{x_a+1}\ketm{x_a \op x_d \op 1}$\\= $-\ketm{x_d \op 1}$}
 & $-\ketm{x_d \op 1}$ \\ \hline
\ $1$ & \makecell{$\ff \ketm{a}$\\+$\ff \ketm{d}$} 
 & \makecell{$\ff (-1)^{x_a}\ketm{a}$ \\+ $\ff (-1)^{x_d}\ketm{d}$} 
 & \makecell{$(-1)^{x_a}\ketm{x_a \op x_d }$\\= $-\ketm{x_d \op 1}$}
 & $-\ketm{x_d \op 1}$ \\ \hline
\end{tabular}
\end{center}
\caption{Evolution of $\ket{\gamma_1}$ and comparison
with $\ket{\mb \beta}=(-1)^{x_b} 
 \ketm{x_b \op x_d}$}
\label{tab:3}

\end{table}

Therefore in all the cases the state post these transformations is 
$$(-1)^{x_b} \ketm{x_b \op x_d} \ket{0}_1\ket{x_b}_2.$$

Now we describe the evolution of the state $\ketm{x_c}\ket{1}_1\ket{x_d}_2$.
As in the previous case, we apply an unitary $\cd U_2$ and then the state
queries to the oracle, which is followed by $\cd \parg{n}{b}{c}$.
We define $U_2$ as the composition of two unitary operators 
defined on $H^q \otimes H_2$, $U_{20}$ and $U_{21}$.
Similar to $U_{10}$ and $U_{11}$, these are operators
that transform the query register depending on $w_2=\ket{0}$
and $\ket{1}$, respectively.
The transformations due to $U_{20}$ and $U_{21}$ are as follows.

\begin{multicols}{2}

\paragraph{$U_{20}$}
\begin{enumerate}

\item $\ketm{0} \rightarrow \ff (\ketm{b}+\ketm{c})$

\item $\ketm{1} \rightarrow \ff (\ketm{b}-\ketm{c})$

\end{enumerate}

\columnbreak

\paragraph{\bf $U_{21}$}
\begin{enumerate}

\item $\ketm{0} \rightarrow \ff (\ketm{b}-\ketm{c})$

\item $\ketm{1} \rightarrow \ff (\ketm{b}+\ketm{c})$

\end{enumerate}

\end{multicols}

That is
\begin{itemize}

\item $\ketm{x_c} \xrightarrow{U_{20}} \ff (\ketm{b}+(-1)^{x_c}\ketm{c}) $

\item $\ketm{x_c} \xrightarrow{U_{21}} \ff (\ketm{b}+(-1)^{x_c+1}\ketm{c}) $

\end{itemize}
The oracle is applied on $\cd U_{21} \cd U_{20} \ket{\gamma_2}$
and on the resultant state, \\ $O_x \cd U_{21} \cd U_{20} \ket{\gamma_2}$  we apply $\cd \parg{n}{b}{c}$.
We observe the evolution for all possible $\{x_b,x_c,x_d \}$ tuples and compare
the final state with $(-1)^{x_b} \ketm{x_b \oplus x_d}\ket{1}_1\ket{x_d}_2$.
We again list solely the evolution of the query register
in Table~\ref{tab:4}, as the other registers
remain unchanged. 

\begin{table}[H]

\begin{center}
\begin{tabular}{ |c|c|c|c|c| } 
 \hline
 $x_d$  & $ \cd U_{21} \cd U_{20}\ket{\gamma_2}$  & $O_x \cd U_{21} \cd U_{20}\ket{\gamma_1}$ 
 & \makecell{$\cd \parg{n}{b}{c} O_x$ \\ $ \cd U_{21} \cd U_{20} \ket{\gamma_2}$} & $\ket{\mb \beta}$ \\ \hline
 $0$ & \makecell{$\ff\ketm{b}$ \\ $+(-1)^{x_c} \ff \ketm{c}$}  
 & \makecell{$\ff (-1)^{x_b}\ketm{b}$   \\  $+\ff (-1)^{2x_c}\ketm{c}$}   
 & \makecell{$(-1)^{x_b}\ketm{x_b}$}
 & $(-1)^{x_b}\ketm{x_b}$ \\ \hline
 $1$ & \makecell{$\ff\ketm{b}$ \\ $+ (-1)^{x_c+1} \ff \ketm{c}$}  
 & \makecell{$\ff (-1)^{x_b}\ketm{b}$   \\  $+\ff (-1)^{2x_c+1}\ketm{c}$}   
 & \makecell{$(-1)^{x_b}\ketm{x_b \op 1}$}
 & $(-1)^{x_b}\ketm{x_b \op 1}$ \\ \hline
\end{tabular}
\end{center}

\caption{Evolution of $\ket{\gamma_2}$ and comparison
with $\ket{\mb \beta}=(-1)^{x_b} 
 \ketm{x_b \op x_d}$}
\label{tab:4} 
\end{table}

Therefore in all the cases the state post these transformations is 
$$(-1)^{x_b} \ketm{x_b \op x_d} \ket{0}_0\ket{x_d}_1.$$
We now look at the collective effect of the transformations 
$\cc U_1$, $\cd U_2$, $O_x$, $\cc \parg{n}{a}{d}$ and $\cd \parg{n}{b}{c}$.
The state at start was 
$$\ff \big(\ketm{x_a}\ket{0}_1\ket{x_b}_2\ket{W_1} + \ketm{x_c}\ket{1}_1\ket{x_d}_2\ket{W_2}   \big).$$
The state after these operations are applied is
$$ \ff \big( (-1)^{x_b}\ketm{x_b \op x_d}\ket{0}_1\ket{x_b}_2\ket{W_1} + (-1)^{x_b}\ketm{x_b \op x_d}\ket{1}_1\ket{x_d}_2\ket{W_2} \big) .$$
We now apply the operations $\cc \cn{Q_n}{w_2}$
followed by $\cn{w_2}{Q_n}$,
evolving the system to 
$
 (-1)^{x_b} \ff \big( \ketm{x_b}\ket{0}_1\ket{x_d}_2 \ket{W_1}+ \ketm{x_b}\ket{1}_1\ket{x_d}_2 \ket{W_2} \big).
$
and this completes the step.
This also shows that for this method the qubit $w_2$ can be swapped with any other work qubit, 
and the method is indifferent towards its choice.
\end{proof}

Observe that this subroutine does not depend on the function 
we are dealing with. However, the advantage is most prominent 
for the classes of functions that we discuss. 
Given the general framework of this technique, it is an interesting problem to check if this technique can have applications in other black box problems in the quantum paradigm as well as if this methodology can be further optimized in the bounded error quantum model.  

Let us now denote the generalized routine 
in this regard that form part of the exact quantum query algorithm.
This is simply obtained by applying the untangling protocol many times, each time 
untangling two new qubits. We omit this proof for brevity. 

\begin{lemma}
\label{lemma:3}
Corresponding to a quantum query algorithm defined on the variables $\xv=(x_1,x_2, \ldots ,x_n)$ with
$\hx$ and $\tx$ are two subspaces that form a disjoint partition of $\xv$,
where $s=2t$
the state
$$ \ket{\beta_0}=
\ff \left( (-1)^{g(\hx)} \ket{\bm{0}}_0 \ket{0}_1 \bigotimes_{i=1}^{s}\ket{x_{r_i}}_{i+1}
+
(-1)^{h(\tx)}\ket{\bm{0}}_0 \ket{1}_1 \bigotimes_{j=1}^{s}\ket{x_{r_{s+j}}}_{j+1} \right) $$
can be evolved to the state $\ket{\beta_f}$ using the protocol ${\sf untangle^s_n}$,
where,

$$\ket{\beta_f}=
\ff \left( (-1)^{g(\hx)} \ket{\bm{0}}_0 \ket{0}_1 
+
(-1)^{h(\tx)}\ket{\bm{0}}_0 \ket{1}_1
\right) 
\bigotimes_{i=1}^{t} \left(\kett*{x_{r(2i)}}_{2i} \kett*{x_{r(s+2i)}}_{2i+1} \right).
$$
by making $t$ queries to the oracle $O_x$.
\end{lemma}

Using this protocol on the state $\ket{\psi_f}$ described in Equation~\eqref{eq:01} 
gives us the state using a further $t$ queries:
$$
\ket{\psi_{end'}}=
\ff \left( (-1)^{g(\hx)} \ket{\pmb{0}}_0 \ket{0}_1 +
(-1)^{h(\tx)} \ket{\pmb{0}}_0 \ket{1}_1 \right)
\bigotimes_{i=1}^{t} \left(\kett*{x_{r(2i)}}_{2i} \kett*{x_{r(s+2i)}}_{2i+1} \right).$$

Applying a Hadamard gate and then measuring $w_1$ in the computational basis gives us the output after a total of $k+ \lceil \frac{s}{2} \rceil$ queries.
The efficiency of the algorithm relies on how well can we partition $\xv$ into $\hx$ and $\tx$ and then choose $k$
and $s$ properly.

Let us now go over the rest of the lemmas, subroutines and intermediate results that we obtain en-route.

\subsubsection{The complete algorithm}
We start with $\ket{\psi}_0$ where $l=k$ and define the 
following transformation. 
\begin{lemma}
\label{lemma:0}
Let $f(\xv)$ be a Boolean function on $n$ variable which is being evaluated 
using an algorithm $\Q(f)$ with the registers $Q_n$ and $k$ qubits of working memory.

Then there exists a transformation $acq(i-1)$ which transforms the state
$\ket{\psi_{i-1}}$ to $\ket{\psi_i}$ by making a single query to the oracle, 
where $\ket{\psi_i}$ is defined as follows.
\
\begin{align*}
\ket{\psi_i}=
\ff \left( \ket{\bm{0}}_0 \ket{0}_1 
\bigotimes_{j=1}^{i-1} \ket{x_j}_{j+1} 
\bigotimes_{j=i+1}^{k} \ket{0}_{j}
+
\ket{\bm{0}}_0 \ket{1}_1 
\bigotimes_{j=1}^{i-1} \ket{x_{k+j}}_{j+1} 
\bigotimes_{j=i+1}^{k} \ket{0}_{j}
\right).
\end{align*}
\end{lemma}
\begin{proof}
This is very similar to the $\sf acquire(i)$ transformation shown in Lemma~\ref{lemma2}, and one can refer to it for a more 
detailed view of a similar process.
Here the difference is that in this case the value of two 
variables are stored in the qubits in the entangled system
with each query.

We show that this transformation can be achieved by using the $acq(i-1)$ transformation defined as the 
sequential application of the following unitaries and the oracle in the given order. 
$$\cc\sg{n}{0}{i}, \cd\sg{n}{0}{k+i} ~ O_x,~ \cc\parg{n}{0}{i}, 
\cd\parg{n}{0}{k+i} , \cn{Q_n}{w_{i+1}}, \cn{w_{i+1}}{Q_n}.$$
That is 
$acq(i-1)=\cn{w_{i+1}}{Q_n}\cn{Q_n}{w_{i+1}}\cd\parg{n}{0}{k+i}
\cc\parg{n}{0}{i}O_x\cd\sg{n}{0}{k+i}\cc\sg{n}{0}{i}$

The step-wise transformation due to $acq(i-1)$ on $\ket{\psi_{i-1}}$ is as follows.
\
\begingroup
\allowdisplaybreaks
\begin{align*}
\ket{\psi_{i-1}}=& \ff \left( \ket{\bm{0}}_0 \ket{0}_1 \ket{x_1}_2 \ldots \ket{x_{i-1}}_{i} \ket{0}_{i+1} \ldots \ket{0}_k \right. 
\\
& \left.  +\ket{\bm{0}}_0 \ket{1}_1 \ket{x_{k+1}}_2 \ldots \ket{x_{i-1}}_{i} \ket{0}_{i+1} \ldots \ket{0}_k 
\right)
\\
 \xrightarrow{\cc\sg{n}{0}{i} ~ \cd\sg{n}{0}{k+i}} & 
 \ff \left(\left(\frac{\ket{\bm{0}}_0+\ket{\bm{i}}_0}{\sqrt{2}} \right) \ket{0}_1 \ket{x_1}_2 \ldots \ket{x_{i-1}}_{i} \ket{0}_{i+1} \ldots \ket{0}_k \right. 
\\
& \left.  + \left(\frac{\ket{\bm{0}}_0 + \ket{\bm{k+i}}_0}{\sqrt{2}} \right) \ket{1}_1 \ket{x_{k+1}}_2 \ldots \ket{x_{i-1}}_{i} \ket{0}_{i+1} \ldots \ket{0}_k \right)
\\
\xrightarrow{O_x} &
 \ff \left(\left(\frac{\ket{\bm{0}}_0+(-1)^{x_i}\ket{\bm{i}}_0}{\sqrt{2}} \right)   \ket{0}_1 \ket{x_1}_2 \ldots \ket{x_{i-1}}_{i} \ket{0}_{i+1} \ldots \ket{0}_k \right. 
\\
+& \left. \left(\frac{\ket{\bm{0}}_0 +(-1)^{x_{k+i}} \ket{\bm{k+i}}_0}{\sqrt{2}} \right) \ket{1}_1 \ket{x_{k+1}}_2 \ldots \ket{x_{i-1}}_{i} \ket{0}_{i+1} \ldots \ket{0}_k \right)
\\
\xrightarrow{\cc\parg{n}{0}{i} ~ \cd\parg{n}{0}{k+i}} &
 \ff \left( \ket{\bm{x_i}}_0 \ket{0}_1 \ket{x_1}_2 \ldots \ket{x_{i-1}}_{i} \ket{0}_{i+1} \ldots \ket{0}_k \right. 
\\
+& \left. \ket{\bm{x_{k+i}}}_0 \ket{1}_1 \ket{x_{k+1}}_2 \ldots \ket{x_{i-1}}_{i} \ket{0}_{i+1} \ldots \ket{0}_k \right)
\\
\xrightarrow{\cn{Q_n}{w_{i+1}}} &
 \ff \left( \ket{\bm{x_i}}_0 \ket{0}_1 \ket{x_1}_2 \ldots \ket{x_{i-1}}_{i} \ket{x_i}_{i+1} \ldots \ket{0}_k \right. 
\\
+& \left. \ket{\bm{x_{k+i}}}_0 \ket{1}_1 \ket{x_{k+1}}_2 \ldots \ket{x_{i-1}}_{i} \ket{x_{k+i}}_{i+1} \ldots \ket{0}_k \right)
\\
\xrightarrow{\cn{w_{i+1}}{Q_n}} &
 \ff \left( \ket{\bm{0}}_0 \ket{0}_1 \ket{x_1}_2 \ldots \ket{x_{i-1}}_{i} \ket{x_i}_{i+1} \ldots \ket{0}_k \right. 
\\
+& \left. \ket{\bm{0}}_0 \ket{1}_1 \ket{x_{k+1}}_2 \ldots \ket{x_{i-1}}_{i} \ket{x_{k+i}}_{i+1} \ldots \ket{0}_k \right)
\\
&= \ket{\psi_i}.
\end{align*}
\
\endgroup
\
\end{proof}

The first transformation is applied $k-1$ times
(which requires $k-1$ queries to be made to the oracle)
and it transforms 
the system to 
$$\ket{\psi_{k-1}}=
\ff \left( \ket{\pmb{0}}_0\ket{0}_1 \otimes _{i=2}^k \ket{x_{i-1}}_i+ 
\ket{\pmb{0}}_0 \ket{1}_1 \otimes _{i=2}^k \ket{x_{k+i-1}}_i \right).$$
Then we apply the following transformational result. 
\begin{lemma}
\label{lemma:onetime}
The state $\ket{\psi_{k-1}}$ can be converted to the state
$$\ket{\beta_0}=\ff \left( (-1)^{\prod\limits_{i=1}^k x_i} \ket{\bm{0}}_0 \ket{0}_1 \bigotimes_{j=2}^k \ket{x_{j-1}}_j
+
(-1)^{\prod\limits_{i=k+1}^n x_i}\ket{\bm{0}}_0 \ket{1}_1 \bigotimes_{j=2}^k \ket{x_{k+j-1}}_j \right)$$
by making a single query to the oracle.
\end{lemma}
\begin{proof}
We begin with the system being in the state
$$\ket{\psi_{\frac{n}{2}-1}}=
 \ff \left( \ket{\bm{0}}_0 \ket{0}_1 \bigotimes_{j=2}^k \ket{x_{j-1}}_j 
+ \ket{\bm{0}}_0 \ket{1}_1 \bigotimes_{j=2}^k \ket{x_{k+j-1}}_j \right).$$
\
At this stage we apply a unitary transformation $C^{k-1}$ 
that changes the state of $Q_n$ 
from $\ket{\bm{0}}_0$ to $\ket{\bm{1}}_0$, controlled on
$w_i=\ket{1}_i, 2 \leq i \leq k$, similar to a $\sf C^{k-1}-NOT$ operation.
That is, iff 
$ \prod_{i=1}^{k-1} x_i=1$, $Q_n$ changes to the state $\ket{\bm{1}}_0$
in the superposition state with $w_1=\ket{0}$ and similarly 
iff 
$ \prod_{i=k+1}^{n-1} x_i=1$ then $Q_n$ changes to $\ket{\bm{1}}_0$
in the superposition state with $w_1=\ket{1}$, forming $\ket{\psi_{k-1}'}$,
which is
$$
C^{k-1} \ket{\psi_{k-1}}=
\ff \left( \kett*{\bm{ \prod_{i=1}^{k-1}x_i }}_0 \ket{0}_1 \bigotimes_{j=2}^k \ket{x_{j-1}}_j
+ \kett*{\bm{\prod_{i=k+1}^{n-1}x_i}}_0 \ket{1}_1 \bigotimes_{j=2}^k \ket{x_{k+j-1}}_j \right).
$$

The next step takes one query and this is the last query the algorithm makes before 
starting the un-entanglement protocol. 
We apply $\cc \pg{n}{\frac{n}{2}} ~\cd \pg{n}{n}$ 
followed by the oracle $O_x$ and then $\cc \pg{n}{\frac{n}{2}} ~\cd \pg{n}{n} $ again.
Let $p_{q}^{r}=\prod\limits_{i=q}^{r}x_i,~ q<r$.
Then the transformation due to the operations is as follows.

\begin{align}
\label{eq:2}
& \ket{\psi_{k-1}'} 
\xrightarrow{\cc \pg{n}{\frac{n}{2}} ~\cd \pg{n}{n}} \\ \nonumber
&
\ff \left( \kett*{\bm{k \times (p_1^{k-1})}}_0 \ket{0}_1 \bigotimes_{j=2}^k \ket{x_{j-1}}_j 
+ \kett*{\bm{n \times (p_{k+1}^{n-1})}}_0 \ket{1}_1 \bigotimes_{j=2}^k \ket{x_{k+j-1}}_j \vph\right)
\\ \nonumber
& \xrightarrow{O_x} 
\ff \left( (-1)^{x_k(p_{1}^{k-1})} \kett*{\bm{k \times (p_{1}^{k-1})}}_0 \ket{0}_1 \bigotimes_{j=2}^k \ket{x_{j-1}}_j  \right. 
\\ \nonumber & 
+ \left. (-1)^{x_n(p_{k+1}^{n-1})}\kett*{\bm{n \times (p_{k+1}^{n-1})}}_0 \ket{1}_1 
\bigotimes_{j=2}^k \ket{x_{k+j-1}}_j \right)
\
\xrightarrow{\cc \pg{n}{\frac{n}{2}} ~\cd \pg{n}{n}} \\ \nonumber  &
\ff \left( (-1)^{x_k(\prod\limits_{j=1}^{k-1}x_j)} \ket{\bm{0}}_0 \ket{0}_1 \bigotimes_{j=2}^k \ket{x_{j-1}}_j 
+ (-1)^{x_n(\prod\limits_{j=1}^{k-1}x_{k+j})}\ket{\bm{0}}_0 \ket{1}_1 \bigotimes_{j=2}^k \ket{x_{k+j-1}}_j \right)
\\ \nonumber
&= 
\ff \left( (-1)^{\prod\limits_{i=1}^k x_i} \ket{\bm{0}}_0 \ket{0}_1 \bigotimes_{j=2}^k \ket{x_{j-1}}_j
+
(-1)^{\prod\limits_{i=k+1}^n x_i}\ket{\bm{0}}_0 \ket{1}_1 \bigotimes_{j=2}^k \ket{x_{k+j-1}}_j \right).
\end{align}
\
\end{proof}

At this stage we only need to untangle the system to obtain the output. 
We now apply Lemma~\ref{lemma:3} to obtain the final output, 
which costs a further $\fl$ queries. 
Together with the result of Lemma~\ref{th:par},
this results in the following Algorithm~\ref{algo}. 

\begingroup
\allowdisplaybreaks


\begin{algorithm}
\caption{$\Q(f)$ to evaluate $f(\xv)=\prod\limits_{i=1}^{\frac{n}{2}}x_i \op 
 \prod\limits_{j=\frac{n}{2}+1}^{n}x_j$ along with query complexity count ($\QC(f)$) :}
\begin{algorithmic}
\label{algo}

\item[1] Begin with the state $\ket{\bm{0}}_0 \otimes_{i=1}^k \ket{0}_i$, consisting of the Query register and $\frac{n}{2}$ work qubits $w_i, 1 \leq i \leq \frac{n}{2}$.

\item[2] We apply a Hadamard to the first work qubit $w_1$ to get
$\ket{\psi_0}=\ff \big( \ket{\bm{0}}_0\ket{0}_1\otimes_{i=2}^k \ket{0}_i 
+\ket{\bm{0}}_0\ket{1}_1 \otimes_{i=2}^k \ket{0}_i \big)$.

\item[3] Then we run the subroutine $acq(i)$ of Lemma~\ref{lemma:0} 
for $\frac{n}{2}-1$ times for $0 \leq i \leq \frac{n}{2}-2$,
which evolves the state from $\ket{\psi_0}$ to $\ket{\psi_{\frac{n}{2}-1}}$,
where 
\begin{align*}
\ket{\psi_i}= & \ff \left( \ket{\bm{0}}_0 \ket{0}_1 
\bigotimes_{j=2}^i \ket{x_{j-1}}_j
\bigotimes_{j=i+1}^k \ket{0}_j 
+\ket{\bm{0}}_0 \ket{1}_1 
\bigotimes_{j=2}^i \ket{x_{k+j-1}}_j
\bigotimes_{j=i+1}^k \ket{0}_j
\right).
\end{align*}

\item[4]
Here let us define $g(\hx)=\prod_{i=1}^{\frac{n}{2}} x_i$
and $h(\tx)=\prod_{j=1}^{\frac{n}{2}} x_{\frac{n}{2}+j}$.
Then we apply the step described in Lemma~\ref{lemma:onetime}
which makes one query to the oracle.

Then after $\frac{n}{2}$ queries the system is in the state
\begin{equation*}
\ket{\psi_{\frac{n}{2}}}=
\ff \left( (-1)^{g(\hx)} \ket{\bm{0}}_0 \ket{0}_1 \bigotimes_{j=2}^k \ket{x_{j-1}}_j
+
(-1)^{h(\tx)}\ket{\bm{0}}_0 \ket{1}_1 \bigotimes_{j=2}^k \ket{x_{k+j-1}}_j \right).
\end{equation*}

\item[5] We then apply the transformation $untangle^s_n$
described in Lemma~\ref{lemma:3} where $s=\frac{n}{2}-1$.
This step requires 
a further $\fl$ queries, and finally the system is in the state
$$\ket{\beta_f}=
(-1)^{g'(\xv)}
\ket{\bm{0}}_0 
\kett*{g(\hx) \op h(\tx) }_1
\bigotimes_{i=1}^{k_1} (\ket{x_{2i}}_{i+1} \ket{x_{k+2i}}_{i+2})
$$
after a total of $\flt$ queries.

\item[6] Get the output by then measuring $w_1$ in the computational basis.

\end{algorithmic}

\end{algorithm}

\endgroup

This coupled with the generalized parity decision tree complexity of the function provides the first separation result.

\begin{theorem}
\label{th:5}
For the function 
$f_1= \prod\limits_{i=1}^{\frac{n}{2}}x_i \op 
 \prod\limits_{j=\frac{n}{2}+1}^{n}x_j $,
we have $\QC(f_1)=\flt$ and $D_{\op}(f_1)=n-1$.
\end{theorem}

\begin{proof}
\ \\~

{\bf $D_{\op}(f)=n-1:$}
This is a direct implication of Proposition~\ref{th:par} where 
$f \in {\sf pdsp}(n,n,2)$

{\bf $\QC(f)=\flt:$}

We run Algorithm~\ref{algo} initializing it in the state
$\ket{\bm{0}}_0 \otimes_{i=1}^k \ket{0}_i$. 
This algorithm makes a total of $\frac{n}{2} + \fl=\flt$ queries, which completes the proof.

\end{proof}

For $f_1$ we are able to separate $\QC(f)$ and $D_{\op}(f)$,
but the algorithm is not provably optimal for this function. 
However we observe that this technique is indeed optimal for the 
following function. 
\begin{corollary}
\label{cor:2}
For the function $f_2$ on $n=2k$ variables where
$f_2(\xv)= \prod\limits_{i=1}^{ \flt }x_i 
\op \prod\limits_{j=\frac{n}{2} +1}^{n} x_j $
we have $\QC(f_2)=Q_E(f_2)=\flt$ and $D_{\op}(f_2)=n-1$.
\end{corollary}

\begin{proof}
\

\paragraph*{$Q_E(f) \geq \flt$}
We can reduce $f_1$ to $\textrm{AND}_{\flt}$ by fixing the variables 
$x_i=0, \flt +1 \leq i \leq n$, and therefore evaluating 
$f$ must take at least $\flt$ queries as we know $Q_E(\textrm{AND}_{\flt})=\flt$.

\paragraph*{$\QC(f)=\flt$}
This function can in fact be written as 

$$f(\xv)= 
\left( \prod\limits_{i=1}^{ \frac{n}{2} }x_i
\bigoplus 
\prod\limits_{i=\flt+1}^{ n }x_i
\right)
\prod\limits_{j=\frac{n}{2}+1}^{ \flt }x_j.
$$

We proceed in the same direction as Theorem~\ref{th:5}. After $\frac{n}{2}$ queries 
the system is in the state
\begin{equation*}
\ket{\psi_{\frac{n}{2}}}=
\ff \left( (-1)^{g(\hx)} \ket{\bm{0}}_0 \ket{0}_1 \bigotimes_{j=2}^k \ket{x_{j-1}}_j
+
(-1)^{h(\tx)}\ket{\bm{0}}_0 \ket{1}_1 \bigotimes_{j=2}^k \ket{x_{k+j-1}}_j \right).
\end{equation*}

We now swap the values of the qubit $w_2$ to $w_k$ so that 
the value $\frac{n}{2}+j$ goes to the qubit $w_{1+2j}$. This ensures 
after the $\sf untangle^s_n$ is applied where $s=\frac{\frac{n}{2}-1}{2}$,
the system's state after $\flt$ queries is

$$
\ket{\bm{0}}_0 \ff \left( (-1)^{g(\hx)}  \ket{0}_1 
+
(-1)^{h(\tx)}\ket{\bm{0}}_0 \ket{1}_1 \right) 
\bigotimes_{j=1}^{s} \ket{x_{j-1}}_{2j} \ket{x_{k+j-1}}_{2j+1} 
$$

We can apply a Hadamard to obtain the state 
$
\ket{\bm{0}}_0 \ket{g(\hx) \op h(\tx)} 
\bigotimes_{j=1}^{s} \ket{x_{j-1}}_{2j} \ket{x_{k+j-1}}_{2j+1}.$

We can now obtain the value of $f(\xv)$ as both the values of 
$
g(\hx) \op h(\tx)=
\left( \prod\limits_{i=1}^{ \frac{n}{2} }x_i
\bigoplus 
\prod\limits_{i=\flt+1}^{ n }x_i
\right)
$
and measuring the qubit $w_{2i+2}$ gives the variable $x_{\frac{n}{2}+i}$. 

\end{proof}

We now briefly describe the phase kickback method before
finally proving the result of Theorem~\ref{th:main}
which gives us the main separation of this paper,
which is just a broader extension of Theorem~\ref{th:5}.

\subsubsection*{Phase Kickback}

The only technique we require to evaluate the functions of this kind 
apart from the ones used in Corollary~\ref{cor:2} 
is that of phase kickback a widely used methodology in black box algorithms.
Suppose we have a $k+1$ qubit system 
$\bigotimes_{i=1}^{k+1}w_i$ in the state
$\left( \bigotimes_{i=1}^k \ket{x_i}_i \right) \ket{-}_{k+1}$.
Let $S \subseteq [k]$ where $[k]= \{1,\ldots,k\}$.
Then the controlled not operation $\sf C^{\abs{S}}-NOT$ controlled on 
$w_i=\ket{1}_i, i \in S$ (i.e. $x_1=i~ \forall i \in S$ ) 
with $w_{k+1}$ being target works as follows. 

\begin{align}
\label{eq:3}
\left( \bigotimes_{i=1}^k \ket{x_i}_i \right) \ket{-}_{k+1}
\xrightarrow{\sf C^{\abs{S}}-NOT}
(-1)^{\left( \prod_{i \in S}x_i \right) }
\left( \bigotimes_{i=1}^k \ket{x_i}_i \right) \ket{-}_{k+1}
\end{align}

Let us now present Theorem~\ref{th:main}, which is one of our main results.
\begin{theorem}
\label{th:main}
Let $f \in \sf pdsp(n,\frt,t+1)$ be a function on $n=2k$ variables
such that
\begin{equation*}
f(\xv)= \left( \prod\limits_{i=1}^{\frac{n}{2}} x_i \bigoplus g(\xv') \right) \left(
\prod\limits_{j=\frac{n}{2}+1}^{\flt} x_j \right)
,~\xv'=\left( x_{\flt+1},x_{\flt+2}, \ldots , x_n \right).
\end{equation*}
where $g$ is perfect direct sum function defined on 
$\left( x_{\flt+1},x_{\flt+2}, \ldots , x_n \right)$ so that it contains $t$ monomials 
such that each monomial consists of at least $t+1$ variables.
Then we have (i) $\QC(f)=Q_E(f)=\flt$, (ii) $D_{\op}(f)=n-t$, (iii) $D(f)=n$.
\end{theorem}
\begin{proof}~

\paragraph*{$Q_E(f) \geq \flt$}
For any such function, if we fix the variables $x_i, \flt+1 \leq i \leq n$ to $0$ 
then the function is reduced to $\sf \textrm{AND}_{\flt}$ which implies $Q_E(f) \geq \flt$.

\paragraph*{$D_{\op}(f)=n-t$}
This is a direct implication of Proposition~\ref{th:par} where the number of monomials 
is $t+1$.

\paragraph*{$\QC(f)=\flt$}
We initialize the algorithm in the state $\ket{\bm{0}}_0 \bigotimes_{i=1}^{k+2} \ket{0}_i$.
We first apply a Not gate and a Hadamard gate to $w_{k+2}$ and then a Hadamard gate 
which evolves the system to
\begin{align*}
\ff
\left(
\ket{\bm{0}}_0 \ket{0}_1 \left( \bigotimes_{i=2}^{k} \ket{0}_i \right) \ket{0}_{k+1} \ket{-}_{k+2}
+
\ket{\bm{0}}_0 \ket{1}_1 \left( \bigotimes_{i=2}^{k} \ket{0}_i \right) \ket{0}_{k+1} \ket{-}_{k+2}
\right).
\end{align*}
This state can be written as $\ket{\psi}_0 \ket{0}_{k+1} \ket{-}_{k+2}$ where
$\ket{\psi_0}$ is the starting state of Theorem~\ref{th:5} and Corollary~\ref{cor:2}.
We now apply the transformations $acq(i), 0 \leq i \leq \frac{n}{2}-2$
as defined in Lemma~\ref{lemma:0} which makes $\frac{n}{2}-1$ queries to the oracle. 
This evolves the system to the state 
$\ket{\psi}_{\frac{n}{2}-1} \ket{0}_{k+1} \ket{-}_{k+2}$, that is
$$
\ff \left(
\ket{\bm{0}}_0 \ket{0}_1 \left( \bigotimes_{i=2}^{k} \ket{x_{i-1}}_i \right) 
+
\ket{\bm{0}}_0 \ket{1}_1 \left( \bigotimes_{i=2}^{k} \ket{x_{k+i-1}}_i \right) \right) \ket{0}_{k+1} \ket{-}_{k+2}
.
$$
This transformation is same as described in Theorem~\ref{th:3} and
since no operation is made on the  $k+1$ and $k+2$-th qubit
their states remain unchanged.
We now acquire the phases 
$(-1)^{\left(\prod_{i=1}^k x_i\right)}$ and $(-1)^{g(\xv')}$.

Since $g$ has a perfect direct sum representation 
there is a single monomial (say $m_1$) in $g(\xv')$ that contains $x_n$.
Let the other variables of the monomial be $x_{\flt+i}, i \in S_1$ 
where $S_1 \subseteq [\fr]$.
Therefore the qubits storing these values in the superposition 
state  with $w_i=\ket{1}_1$ are $w_{(1+\flt+i)}, i \in S_1$.
We then apply the following operations.

\begin{itemize}
\item[] Controlled on $w_1=\ket{0}$, we apply a controlled not gate with 
$Q_n$ as target and $w_i, 2 \leq i \leq k$ as controls.

\item[] Controlled on $w_1=\ket{1}$, we apply a controlled not gate with 
$Q_n$ as target and $w_{\fl+i+1}, i \in S_1$ as controls.
\end{itemize}
This transforms the system to the state
\begin{align*}
 \ff
\left( 
\kett*{\bm{ \prod_{i=1}^{k-1} x_i}}_0 \ket{0}_1 \left( \bigotimes_{i=2}^{k} \ket{x_{i-1}}_i \right)
\right.
+
\left.
\kett*{\bm{\prod_{i \in S_1} x_{\flt+i}}}_0 \ket{1}_1 \left( \bigotimes_{i=2}^{k} \ket{x_{k+i-1}}_i \right) 
\right) \\
\ket{0}_{k+1} \ket{-}_{k+2}
.
\end{align*}
\
The next operations are $\cc P^n_k$ and $\cd P^n_n$, followed by the oracle 
and then $\cc P^n_k$ and $\cd P^n_n$ again, which results in the same transformation
as shown in Equation~\ref{eq:2} with the only difference that the monomial corresponding to
the superposition state with $w_1=\ket{1}$ has changed.
This forms 
\begin{align*}
\ket{\psi_k}=
\ff 
\left( (-1)^{ \prod_{i=1}^{k} x_i} 
\ket{\bm{0}}_0 \ket{0}_1 \left( \bigotimes_{i=2}^{k} \ket{x_{i-1}}_i \right) 
\right.
+
\left. (-1)^{m_1}
\ket{0}_0 \ket{1}_1 \left( \bigotimes_{i=2}^{k} \ket{x_{k+i-1}}_i \right) 
\right)
\ket{0}_{k+1} \ket{-}_{k+2}
\end{align*}
after $\frac{n}{2}$ queries.

We now obtain the phases corresponding to the other monomials $m_i, 2 \leq i \leq t$ 
using phase kickback as shown in  Equation~\ref{eq:3}.
Let the variables in the $i$-th monomial be $x_{\flt+j}, j \in S_i$,$ S_i \subseteq [\fr]$.
Controlled on $w_1=\ket{1}$,
corresponding to each monomial $m_i$, we apply 
the operation $\sf C^{\abs{S_i}}-NOT$ on $w_{k+2}$,
where the $\abs{S_i}$ controls are $w_{\fl+j+1}=\ket{1}, j \in S_i$.
After the phases corresponding to the 
$t-1$ monomials of $g$ are evaluated this way, the system is in the 
state
\begin{align*}
\ket{\psi_k}= \ff & \left( (-1)^{ \prod_{i=1}^{k} x_i} 
\ket{\bm{0}}_0 \ket{0}_1 \left( \otimes_{i=2}^{k} \ket{x_{i-1}}_i \right) 
\right.
\\+ &
\left. (-1)^{\op_{i=1}^t m_i}
\ket{0}_0 \ket{1}_1 \left( \otimes_{i=2}^{k} \ket{x_{k+i-1}}_i \right) 
\right)
\ket{0}_{k+1} \ket{-}_{k+2}
\\ =
\ff & \left( (-1)^{ \prod_{i=1}^{k} x_i} 
\ket{\bm{0}}_0 \ket{0}_1 \left( \otimes_{i=2}^{k} \ket{x_{i-1}}_i \right)
\right.
\\+ &
\left. (-1)^{g(\xv')}
\ket{0}_0 \ket{1}_1 \left( \otimes_{i=2}^{k} \ket{x_{k+i-1}}_i \right) 
\right).
\ket{0}_{k+1} \ket{-}_{k+2}
\end{align*}

From here on the algorithm proceeds identically as Corollary~\ref{cor:2}
We first swap the values of the qubits in the superposition state with $w_1=\ket{1}$
so that the qubits are in the state $w_{2i+1}=\ket{x_{k+i}}$.
Then the untangling protocol makes $\frac{s}{2}$ queries and the 
system is in the following state after an application of Hadamard gate on $w_1$.

$$
(-1)^{g'(\xv)}
\ket{\bm{0}}_0 
\kett*{\prod_{i=1}^k x_i \op g(\xv') }_1
\bigotimes_{i=1}^{s} (\ket{x_{2i}}_{2i} \ket{x_{k+i}}_{2i+1})  
\ket{0}_{k+1}\ket{-}_{k+2}.
$$

From here-on we can obtain the value of the monomial 
$\prod\limits_{j=\frac{n}{2}+1}^{\flt} x_j $ as the value of each variable $x_j$
is the state of the qubit $w_{2j+1}$, which is in the state $\ket{x_j}_{2j+1}$.
This completes the proof.

\end{proof}

The number of functions covered by the class, referred in Theorem~\ref{th:main}, is as follows. 
\begin{corollary}
\label{cor:3}
For any $n$ there are $\Omega \left( 2^{\frac{\sqrt{n}}{2}} \right)$ functions
(without considering permutation of variables)
for which we can obtain $\Q(f)=Q_E(f)<D_{\op}(f)$.
\end{corollary}

\begin{proof}

We give a lower bound on number of functions which satisfy the 
constraints of the function described in Theorem~\ref{th:main}.
Without considering the permutation of variables, we can 
simply count the number of ways the function $g(\xv')$
can be constructed. 
The function $g$ is defined on $\fr$ variables and 
is it self a perfect direct sum function as defined in Definition~\ref{def:2}.
If $g$ contains $t$ monomials then then each of the monomial must 
have at least $t+1$ variables in them. This is because
$\prod_{i=1}^k x_i \bigoplus g(\xv')$ must satisfy the constraints 
of Definition~\ref{def:2}.
Therefore each construction of $g$ is a different way of 
partitioning the $\fr$ variables into $t$ sets.
If we do not consider which variable is in which monomial,
and rather just the distribution of the variables in the
partitions, then this becomes same as
finding the number of solutions to 
$\sum_{i=1}^t v_i=\fr$ where $v_i \geq t+1~ \forall i$.
We do this is as it is well known that if a function is 
derived from some other function just by a permutation of 
the variable names, they have the same query complexity
and are called PNP equivalent~\cite{exact}. We aim to count the 
functions that are not PNP equivalent with each other.
The number of such partitions is $\displaystyle {n +t -(t+1)^2 -1 \choose t-1}= {\fr-t^2-t-1 \choose t-1}$.
Here $t$ is minimum $1$ and at maximum $\sqrt{\fr-1}$.
Therefore the total possible number of function is
\begin{align*}
\left(
 \displaystyle \sum_{x=1}^{\sqrt{\fr-1}} {\fl-x^2-x -1 \choose x-1}
\right)
> 
\left(
 \displaystyle \sum_{x=1}^{\sqrt{\fr-1}} { \sqrt{\fr-1}  \choose x}
\right)
=
\Omega \left( 2^{\sqrt{\frac{n}{4}}} \right)
=
\Omega \left( 2^{\frac{\sqrt{n}}{2}} \right)
.
\end{align*}
\
\end{proof}

Again, note that the advantage is from being able to deterministically untangle two qubits with a single query,
owing to the result of Theorem~\ref{lemma:1} and the fact that these functions have high granularity.

The next important fact is in untangling we have a degree 
of freedom in terms of which variables we want to carry over 
to the end, and then their values can again be deterministically 
to obtain other monomials. 
In fact in the state $\ket{\beta_0}$, if there are $s$ 
variables each whose values are stored in the two superposition 
state, we can carry over $\lceil \frac{s}{2} \rceil$ values
from each of the superposition states to the final state 
that is simply a tensor product of qubits in computational 
basis states, meaning the value of all the variables stored 
in the working memory can be deterministically obtained. 

This is evident from the structure of the state that we 
obtain at the end for $f_1(\xv)$ and $f_2(\xv)$
(For $f_2$ we have already decided on which values
from $x_i, \frac{n}{2} \leq i \leq n$ we want to carry over
to the final, pre-measurement state):  
$$\ket{\beta_f}=
(-1)^{g'(\xv)}
\ket{\bm{0}}_0 
\kett*{\prod\limits_{i=1}^k x_i \op \prod\limits_{i=k+1}^n x_i }_1
\bigotimes_{i=1}^{k_1} (\ket{x_{r(2i)}}_{i+1} \ket{x_{k+i}}_{i+2})  
.$$
The algorithm for the other 
functions progresses in a similar manner.

This coupled with the fact that the $\sf pdsp$ class has high granularity, which allows us to efficiently 
lower bound the generalized parity tree complexity gives us the desired advantage. 

\section{The results for MM type Bent functions}
\label{sec:4}
In this section we apply our techniques on Maiorana-McFarland (MM) type functions~\cite{bent1}. As we have pointed out in the 
introduction, our investigation started with the study of MM bent functions on small number of variables.

\subsection{MM Bent Functions on $4$ and $6$ variables}
\label{sec:31}
It has been shown in \cite{parity} that we can construct an exact quantum query algorithm 
to evaluate any MM Bent function $f$ on $n$ variables with $\lceil \frac{3n}{4} \rceil$ queries using the parity decision tree technique.
This method utilizes the definition of the MM construction.
However, given that we only know $Q_E(f) \geq \frac{n}{2}$,
this does not rule out the possibility of an algorithm with lesser query complexity.
To verify the tightness of the upper bound due to the parity method,
we obtained the exact quantum query complexity of the functions
$f^{id}_4(x)=x_1x_3 \op x_2x_4$ and $f^{id}_6= x_1x_4 \op x_2x_5 \op x_3x_6$ 
using the semidefinite programming method of~\cite{sdp}, utilizing the CVX package of Matlab~\cite{cvx}.
It is worth mentioning here that the default solvers of CVX couldn't accurately solve the 
semidefinite program for $n=6$, and we had to use a commercial solver called 'Mosek'.
The parity method requires $3$ and $5$ queries for $f^{id}_4$ and $f^{id}_6$ respectively. 
For $f^{id}_4$ the exact quantum query complexity of the function indeed matched that value.
However, we found $Q_E(f^{id}_6)=4$, which is lower than the query complexity of the parity method
and could not formulate any other parity based method that reached the query complexity of $4$. This was the starting point of trying to design a new exact quantum query algorithm that could meet touch this lower bound. Although $\F$-polynomial and untangling based algorithms that we designed are not provably optimal for this class, we were able to use the same philosophy to obtain optimal results for the $\sf pdsp$ class, as described in Section~\ref{sec:2}. In this direction we first develop our algorithm for the function $f^{id}_n$ 
and then extend it for a larger class of MM type bent functions, 
of the size doubly exponential in $\frac{n}{4}$. 

\subsection{Extending for general $n$} 
\label{sec:41}
We first extend our techniques to build an exact quantum algorithm 
for evaluating $f^{id}_n(\xv)=
\bigoplus_{i=1}^{\frac{n}{2}} \left( x_ix_{\frac{n}{2}+i} \right)$ 
that requires $\frac{n}{2} + \lceil \frac{n}{8} \rceil = \lceil \frac{5n}{8} \rceil$ 
queries. We then observe that the same algorithm in fact evaluates a much larger
class of functions in $\mathbb{B}_n$, although the permutation is still identity
permutation. Finally we show how 
this algorithm can be modified
to evaluate functions in $\mathbb{B}_n$ beyond the identity permutation. 

For the functions $f^{id}_n$ we need $l+1$ qubits as working memory where $l=\fl$. 
First we describe the phase obtaining method corresponding to the monomials when evaluating $f^{id}_n$. This is very similar to that of the $\sf pdsp$ class.

\begin{lemma}
\label{lemma2}
Corresponding to a query algorithm for a function on $n=2k$ variables
with $l= \fl$,
the state $\ket{\psi_i}$ can be transformed to the state $\ket{\psi_{i+1}}$
by making two queries to the oracle for $0 \leq i < l$ where the state 
$\ket{\psi_i}$ is defined as
\
\begin{align*}
\ket{\psi_i}=&\ff(-1)^{\sum_{j=1}^i x_jx_{k+j}} \ket{\mb{0}}_0\ket{0}_1  
\bigotimes_{j=1}^i \ket{x_j}_{j+1}
\bigotimes_{j=i+2}^{l+1} \ket{0}_j \\
+&\ff (-1)^{ \sum_{j=1}^i x_{l+j}x_{l+k+j}} \ket{\mb{0}}_0\ket{1}_1
\bigotimes_{j=1}^i \ket{x_{l+j}}_{j+1}
\bigotimes_{j=i+1}^{l} \ket{0}_{j+1} 
\end{align*}

\end{lemma}

\begin{proof}

We define a protocol $\sf acq_1(i)$ which functions as follows.
We first apply the unitaries $\cc \sg{n}{0}{i+1}$ and $\cd \sg{n}{0}{l+i+1}$ on $\ket{\psi}_i$. 
This transforms the system to 
\
\begin{align*}
&\ff\left( (-1)^{\sum_{j=1}^i x_jx_{k+j}} \right) \frac{\ket{\mb{0}}_0 + \ket{\mb{i+1}}_0}{2}\ket{0}_1\ket{x_1}_2\ldots\ket{x_i}_{i+1}\ket{0}_{i+2}\ldots \ket{0}_{l+1} 
\\
+&\ff \left( (-1)^{ \sum_{j=1}^i x_{l+j}x_{l+k+j}} \right)  \frac{\ket{\mb{0}}_0 + \ket{\mb{l+i+1}}_0}{2}\ket{1}_1\ket{x_{l+1}}_2 \\ & \ldots\ket{x_{l+i}}_{i+1}\ket{0}_{i+2}\ldots \ket{0}_{l+1}. &
\end{align*}
We apply the oracle on this state, forming
\begin{align*}
&\ff \left( (-1)^{\sum_{j=1}^i x_jx_{k+j}} \right)
\frac{\ket{\mb{0}}_0 + (-1)^{x_{i+1}}\ket{\mb{i+1}}_0}{2}\ket{0}_1\ket{x_1}_2 \\ &
\ldots\ket{x_i}_{i+1}\ket{0}_{i+2}\ldots \ket{0}_{l+1}
\\
+& \ff \left( (-1)^{ \sum_{j=1}^i x_{l+j}x_{l+k+j}} \right) \frac{\ket{\mb{0}}_0 + (-1)^{x_{l+i+1}}\ket{\mb{l+i+1}}_0}{2}\ket{1}_1\ket{x_{l+1}}_2 \\ &
\ldots\ket{x_{l+i}}_{i+1}  \ket{0}_{i+2}\ldots \ket{0}_{l+1}.
\end{align*}
The next unitaries are $\cc \parg{n}{0}{i+1}$ and $\cd \parg{n}{0}{l+i+1}$, which forms the state 
\begin{align*}
&\ff(-1)^{\sum_{j=1}^i x_jx_{k+j}} \ket{\mb{x_{i+1}}}_0 \ket{0}_1\ket{x_1}_2\ldots\ket{x_i}_{i+1}\ket{0}_{i+2}\ldots \ket{0}_l \\
+&\ff (-1)^{ \sum_{j=1}^i x_{l+j}x_{l+k+j}} \ket{\mb{x_{l+i+1}}}_0 \ket{1}_1\ket{x_{l+1}}_2\ldots\ket{x_{l+i}}_{i+1}\ket{0}_{i+2}\ldots \ket{0}_{l+1}.
\end{align*}

We then use the permutation matrices $\cc \pg{n}{k+i+1}$ and $\cd \pg{n}{l+k+i+1}$, 
then make a query to the oracle, and use the permutation matrices
with the same controls. The resultant state is then 
\begin{align*}
&\ff(-1)^{\sum_{j=1}^i x_jx_{k+j}} (-1)^{x_{i+1}x_{k+i+1}} \ket{\mb{x_{i+1}}}_0  \ket{0}_1\ket{x_1}_2\ldots\ket{x_i}_{i+1}\ket{0}_{i+2}\ldots \ket{0}_{l+1} \\
+&\ff (-1)^{ \sum_{j=1}^i x_{l+j}x_{l+k+j}} (-1)^{x_{i+l+1}x_{l+k+i+1}} \ket{\mb{x_{l+i+1}}}_0 \ket{1}_1\ket{x_{l+1}}_2\ldots\ket{x_{l+i}}_{i+1} \\ 
&\ket{0}_{i+2}\ldots \ket{0}_{l+1}. \\
=&\ff(-1)^{\sum_{j=1}^{i+1} x_jx_{k+j}} \ket{\mb{x_{i+1}}}_0  \ket{0}_1\ket{x_1}_2\ldots\ket{x_i}_{i+1}\ket{0}_{i+2}\ldots \ket{0}_{l+1} \\
+&\ff (-1)^{ \sum_{j=1}^{i+1} x_{l+j}x_{l+k+j}} \ket{\mb{x_{l+i+1}}}_0 \ket{1}_1\ket{x_{l+1}}_2\ldots\ket{x_{l+i}}_{i+1}\ket{0}_{i+2}\ldots \ket{0}_{l+1}. \\
\end{align*}

Finally we swap the values of the query register and the $i+2$-th work qubit which is i the $\ket{0}$ state
in both superposition state. Which results in the state
\begin{align*}
&
\ff(-1)^{\sum_{j=1}^{i+1} x_jx_{k+j}} \ket{\mb{0}}_0\ket{0}_1
\bigotimes_{j=1}^{i+1} \ket{x_j}_{j+1}
\bigotimes_{j=i+2}^{l+1} \ket{0}_{j_1} \\
+&\ff (-1)^{ \sum_{j=1}^{i+1} x_{l+j}x_{l+k+j}} \ket{\mb{0}}_0\ket{1}_1
\bigotimes_{j=1}^{i+1} \ket{x_{l+j}}_{j+1}
\bigotimes_{j=i+2}^{l} \ket{0}_{j+1} \\
&=\ket{\psi_{i+1}}
\end{align*}
Therefore we get $\ket{\psi_i} \xrightarrow{\sf acq_1(i)} \ket{\psi_{i+1}}$
and this completes the proof.

\end{proof}

We now start describing the algorithm for evaluating $f^{id}_n$ 
assuming $n \equiv 0 \mod 4$ ( then $l=\fl=\fs$).
We start with the state 
$$\ket{\psi}_0= \ff \big( \ket{\mb{0}}_0\ket{0}_1 \bigotimes_{j=1}^l \ket{0}_{j+1}
+
\ket{\mb{0}}_0\ket{1}_1 \bigotimes_{j=1}^l \ket{0}_{j+1}  \big)$$

We then apply $\sf acq_1(i), {0 \leq i < l}$ described in Lemma~\ref{lemma2}.
transforming the system to the state
\begin{align*}
\ket{\psi}_{l}=&\ff(-1)^{\sum_{j=1}^{l} x_jx_{k+j}} \ket{\mb{0}}_0\ket{0}_1
\bigotimes_{j=1}^l \ket{x_j}_{j+1} \\
+&\ff (-1)^{ \sum_{j=1}^{l} x_{l+j}x_{l+k+j}} \ket{\mb{0}}_0\ket{1}_1
\bigotimes_{j=1}^l \ket{x_{l+j}}_{j+1}
\end{align*}

At this point, we have used $\frac{n}{2}$ queries and have obtained all the required phases 
needed to evaluate the function's value for the corresponding input. 
However, the system at this point is entangled if any of the qubits $w_i$ is in different state
in the two states.
We now construct the next building block of the algorithm, which is built to bring two qubits 
in the two superposition states (one with $w_1=\ket{0}$ and the other with $w_1=\ket{1}$
to the same state using a single query.
This method of untangling two qubits with one query is the foundational step of algorithm. 
\
Since there are $\frac{n}{4}$ qubits containing values of different variables 
in $\ket{\psi}_{\frac{n}{4}}$, this process needs to be applied 
$\lceil \frac{n}{8} \rceil$ times to get the desired result and un-entangle the system.
We now describe this methodology.

Recall that the state of the system after
$\frac{n}{2}$ queries is in the state
\begin{align*}
\ket{\psi}_{\fs}=&\ff(-1)^{ \left( \sum_{j=1}^{\frac{n}{4}} x_jx_{\frac{n}{2}+j} \right) } 
\ket{\mb{0}}_0\ket{0}_1\ket{x_1}_2\ket{x_2}_3\ldots \ket{x_{\frac{n}{4}}}_{\frac{n}{4}+1} 
\\
+&\ff (-1)^{ \left( \sum_{j=1}^{\frac{n}{4}} x_{\frac{n}{4}+j}x_{\frac{3n}{4}+j} \right) } 
\ket{\mb{0}}_0\ket{1}_1\ket{x_{\fs+1}}_2\ket{x_{\fs+2}}_3\ldots \ket{x_{\frac{n}{2}}}_{\frac{n}{4}+1}
\end{align*}
We define 
$f_1(\xv)=\left( \bigoplus_{j=1}^{\fl} x_jx_{\frac{n}{2}+j} \right)$ and
$f_2(\xv)= \left( \bigoplus_{j=1}^{\fl} x_{\fl+j}x_{\flt+j} \right)$, and thus
$f^{id}_n(\xv)= f_1(\xv) \op f_2(\xv)$ when $n \equiv 0 \mod 4$.
Therefore,
\begin{align}
\label{eq:psin4}
\ket{\psi}_{\fs}=\ff \left( (-1)^{ f_1(\xv) } 
\ket{\mb{0}}_0\ket{0}_1 
\bigotimes_{j=1}^{\fl} \ket{x_j}_{j+1}
+ (-1)^{ f_2(\xv) } 
\ket{\mb{0}}_0\ket{1}_1
\bigotimes_{j=1}^l \kett*{x_{\fl+j}}_{j+1}
\right).
\end{align}
We have the acquired the value of $f_1$ and $f_2$ as local phases.
At this stage we use the protocol $\sf untangle$ defined in
Theorem~\ref{lemma:1} iteratively to un-entangle the state. 
We define this protocol in a generic manner so that it can be used for 
other states as well.

We now directly apply the $\sf untangle^s_n$ protocol to the state
$\ket{\psi}_{\fs}$ as described in Equation~\ref{eq:psin4}
as
$$\ket{\psi}_{\fs}=
\ff \left( (-1)^{ f_1(\xv) } 
\ket{\mb{0}}_0\ket{0}_1 
\bigotimes_{j=1}^{\fs} \ket{x_j}_{j+1}
+ (-1)^{ f_2(\xv) } 
\ket{\mb{0}}_0\ket{1}_1
\bigotimes_{j=1}^{\fs} \ket{x_{\fs+j}}_{j+1}
\right).
$$

Here $\ket{\psi_{\fs}}$ can be described as the state $\ket{\beta_0}$
described in Lemma~\ref{lemma:3}
by putting $s=\fs$ and $r_i=i, 1 \leq i \leq \fs$.
Thus $\sf untangle^{s}_n$ takes $\lceil \frac{\fs}{2} \rceil = \lceil \frac{n}{8} \rceil$ queries and the system is in the state
\
\begin{align*}
(-1)^{g'(\xv)} 
\ketm{0}  \kett*{x_1x_{\frac{n}{2}+1} \op x_2x_{\frac{n}{2}+2} \ldots \op x_{\frac{n}{2}}x_n }_1 
\ket{x_2}_2 \ket{x_{l+2}}_3 \ldots \ket{x_{\frac{n}{2}}}_{\frac{n}{4}+1}
\end{align*}
and measuring $w_1$ in the computational state gives us the output. 

\ \\
In case of $n \equiv 2 \mod 4$ the number of monomials is not even.
and we have $f^{id}_n(\xv)=f_1(\xv) \op f_2(\xv) \op x_{\frac{n}{2}}x_n$.
Thus we acquire the 
phase related to $\fl$ monomials in the state with $\ket{w_1}=0$ and for $\fr$ monomials
in the state with $\ket{w_1}=1$.
We apply $\sf acquire(i), 0 \leq i \leq \fl$, 
bringing the system to the state
$$
\ff \left( 
(-1)^{f_1(\xv)} 
\ketm{0}\ket{0}_1 
\bigotimes_{j=1}^{\fl} \ket{x_j}_{j+1}
+(-1)^{f_2(\xv)} 
\ketm{0}\ket{1}_1
\bigotimes_{j=1}^{\fl} \ket{x_{\fl+j}}_{j+1}
\right).
$$
\
Thus, the monomial $x_{\frac{n}{2}}x_n$ still needs to be evaluated.
At this point we obtain the last monomial with the state containing $\ket{1}_1$ using two queries.
For the superposition state with $\ket{0}_1$,
the value of qubit $\ket{x_1}_2$ is transformed to
$\ket{x_{\fl+1}}_2$ and the query register holds the value of $x_{\frac{n}{2}}$.
Thus after $\fl+2$ queries the system is in the state
\begin{align*}
(-1)^{g'(\xv)}&
\left(
\ff(-1)^{f_1(\xv)} 
\ketm{x_{\frac{n}{2}}}\ket{0}_1\ket{x_{\fl+1}}_2
\bigotimes_{j=2}^{\fl} \ket{x_j}_{j+1} \right.
\\+& 
\left.
\ff (-1)^{ f_2(\xv) \op x_{\frac{n}{2}}x_n } 
\ketm{x_{\frac{n}{2}}}\ket{1}_1\ket{x_{\fl+1}}_2
\bigotimes_{j=2}^{\fl} \ket{x_{\fl+j}}_{j+1} \right)
\end{align*}

Thus at this stage apart from $w_1$, $\fr-2$ qubits have different variables in the two superposition state.
We swap the value of $Q_n$ and $w_2$ with $w_{\fl}$
and $w_{\fl+1}$ and apply the $\sf untangle^s_n$ protocol 
(with $s=\fr-2$) and reverse the swap operations.
This protocol makes $\lceil \frac{\fr-2}{2} \rceil$ queries
and the system is in the state
$$ 
(-1)^{g'(\xv)}
\ketm{x_{\frac{n}{2}}} \ket{x_1x_{\frac{n}{2}+1} \op x_2x_{\frac{n}{2}+2} \ldots \op x_{\frac{n}{2}}x_n}_1 
\ket{x_{\fl+1}}_2
\ket{x_3}_3 \ldots \ket{x_{\frac{n}{2}-1}}_{\fl+1}.
$$
\
This gives us the answer after making a total of $2\fl +2 +\lceil \frac{\fr}{2} \rceil -1 = \lceil \frac{5n}{8} \rceil$ 
queries to the oracle.
Thus both in the case of $n \equiv 0 \mod 4$ and $n \equiv 2 \mod 4$ we require $\lceil \frac{5n}{8} \rceil$ queries
to evaluate the function $f^{id}_n$ exactly.

The above discussion, combined with the Lemma~\ref{lemma2}, 
Theorem~\ref{lemma:1}, Lemma~\ref{lemma:3} 
and the description of Algorithm~\ref{algo} can be summarized as the following
theorem.
\begin{theorem}
\label{th:3}
The function $f^{id}_n$ can be evaluated by an exact quantum algorithm 
that makes $\lceil \frac{5n}{8} \rceil$ queries to the oracle
and uses $\fl+1$ qubits as working memory.
\end{theorem}

This completes the description of the exact quantum algorithm that evaluates 
$f^{id}_n$ using $\lceil \frac{5n}{8} \rceil$ queries. 
As we can observe, in case of $n \equiv 0 \mod 4$ 
the qubits $w_2, w_3, \dots w_{\fl +1}$ are in the states
$x_2, x_{\fl +2},x_4,x_{\fl +2}, \ldots, x_{\frac{n}{2}}$ respectively.
If $n \equiv 2 \mod 4$ then the query register contains the variable $x_{\frac{n}{2}}$
and the qubits contain the variables 
$x_{\fl+1}, x_2, x_{\fl+2}$ so on.
In both cases value of $\fr$ input variables is obtained via these qubits.
Therefore we can also evaluate any function $g$ depending on these variables 
without making any more queries to the oracle, which we summarize in the following corollary.

\begin{corollary}
This algorithm can also be used to evaluate any MM type Bent function with identity permutation
and the function $g$ having at most $\fr$ influencing variables. 
\end{corollary}

\subsection{Beyond the Identity Permutation}
We have shown that our algorithm can evaluate the MM Bent functions
of type $f^{id}_n \op g(x')$ where $x'$ is a subset of $\hat{x}$ consisting of 
at most $\fr$ variables.
However, the techniques we have used are do not restrict the permutation to be identity
permutation. The algorithm works on dividing the variables of $\hat{x}$ into two (close to)
equal disjoint sets and then calculating the value of the corresponding points of $\tilde{x}$,
depending on the permutation.
In case of the identity permutation, since the variable $x_{\frac{n}{2}+i} \in \tilde{x}$
depended solely on the value of $x_i \in \hat{x}$ we could realize this procedure in a 
sequential manner.
Therefore, as long we have a permutation such that it can be expressed as the concatenation 
of two permutations on $\frac{n}{4}$ variables each,
or more precisely concatenation of permutations on $\fl$ and $\fr$ variables, 
we should be able to calculate the influencing variables in $\tilde{x}$ corresponding 
to the values of the variables in $\hat{x}$ at parallel, 
and thus be able to evaluate the function with the same query complexity of $\lceil \frac{5n}{8} \rceil$.
We now concretize this relaxation in restraint and the corresponding modifications needed 
in the algorithm.

\begin{theorem}
\label{th:4}
Let $f$ be an MM Bent function $f$ on $n$ variables such that $f= \phi(\hat{x}) \cdot \tilde{x} \op g(x')$,
with the following constraints:

\begin{enumerate}

\item[1] $\phi_1$ and $\phi_2$ are two permutations such that
$\phi(\hat{x}) \cdot \tilde{x} = \phi_1(\hat{y}) \cdot \tilde{y} \op \phi_2(\hat{z}) \cdot \tilde{z}$
\item[2] The sets of variables $\hat{y},\hat{z},\tilde{y},\tilde{z}$ are all disjoint,
$\abs{\hat{y}}=\abs{\tilde{y}}= \fl$, $\abs{\hat{z}}=\abs{\tilde{z}}= \fr$
\item[3] $\hat{y} \cup \hat{z} =\hat{x}$ and $\tilde{y} \cup \tilde{z} = \tilde{x}$
\item[4] $x' \subset{\hat{x}}, \abs{x' \cap \hat{y}}  \leq \lceil \frac{n}{8} \rceil, 
\abs{x' \cap \hat{z}}  \leq \lceil \frac{n}{8} \rceil$
\end{enumerate}

Then the function can be evaluated by an exact quantum query algorithm 
that makes $\lceil \frac{5n}{8} \rceil$ queries to the oracle and uses $\frac{n}{2}+1$ qubits as working memory.
\end{theorem}

\begin{proof}
Let the variables of $\hat{y}$ be $x_{i_1},x_{i_2}, \ldots x_{i_{\fl}}$ and 
$x_{i_{\fl+1}},x_{i_{\fl+2}}, \ldots x_{i_{\frac{n}{2}}}$ be the variables of $\hat{z}$. 
We start the system in the state in the all zero state and apply a Hadamard gate on the qubit $w_1$
to get the state 
$$\ff \ketm{0}\ket{0}_1\ket{0}_2\ldots\ket{0}_{\frac{n}{2}+1} + 
\ff \ketm{0}\ket{1}_1\ket{0}_2\ldots\ket{0}_{\frac{n}{2}+1}.$$

Corresponding to the state with $w_1=\ket{0}$,
the algorithm progresses as follows:
we obtain the values of the $\fl$ variables in 
$\hat{y}$ using the first $\fl$ queries and 
store them in the qubits $w_2,w_3, \ldots w_{\fl +1}$.
Before the $t$-th query, where $1 \leq t \leq \fl$ the gate $\cc \parg{n}{0}{i_t}$,
is applied, followed by the oracle and then the value of 
query register is swapped with the $t+1$-th work qubit, 
which is in the state $\ket{0}$ at this point. 
The next $\fl$ queries are used to obtain the corresponding 
linear function in $\tilde{y}$ as follows.
The linear function in $\tilde{y}$ can be encoded using $2^{\fl}$ unitary operations.
Each operation correspond to a point in $\hat{y}$. 
For example, if $\phi_1(e_1,e_2,\ldots e_{\fl})=(h_1,h_2,\ldots h_{\fl}),
e_t,h_t \in \{0,1\}$
then we apply a multiple target $\sf C^{\fl}-NOT$ operation controlled on $w_1=\ket{0}$,$w_2=\ket{e_1},\ldots, w_{\fl+1}=\ket{e_{\fl}}$,
with the targets being the qubits $w_{\fr+1+t}$ where $h_t=1$.
We apply these kinds of operations for all $2^{\fl}$ points in $\hat{y}$.
Note that for any input only one of these operations will have all controls satisfied.
Once this operation is applied, 
we have the indexes of the variables in $\tilde{y}$ obtained which are influential at that input point.
We can obtain the corresponding phase one after another in a multiplied form 
by putting a C-NOT from the qubit $w_{\fr+1+t}$ to the query register and
then apply the appropriate $\cc \pg{n}{v}$
gate where $v$ depends on the encoding used.
This is followed by a query to the oracle and then the 
C-NOT operation is applied again to un-compute the
query register.
Thus, after $2 \times \fl$ query this superposition state is in the form.
At this point we apply the $\sf C^{\fl}-NOT$ operations to un-compute
the garbage in the qubits $w_{\fr+1}$ to $w_{\frac{n}{2}+1}$,
leading the system to
$$  (-1)^{\phi_1{\hat{y}} \cdot \tilde{y}} \ketm{0}\ket{0}_1 \ket{x_{i_1}}_2\ldots \ket{x_{i_{\fl}}}_{\fl+1} \ket{0}_{\fl+2} \ldots \ket{0}_{\frac{n}{2}+1}$$

Similarly, in case of the state with $w_1=\ket{1}$,
this set of operations take $2 \times \fr$ queries to get the
phase $(-1)^{\phi_2(\hat{z}).\tilde{z}}$
and at this state the superposition state is in

$$  (-1)^{\phi_1{\hat{z}} \cdot \tilde{z}} \ketm{0}\ket{1}_1 \ket{x_{i_{\fl+1}}}_2\ldots \ket{x_{i_{\fl+\fr}}}_{\fr+1} \ket{0}_{\fr+2} \ldots \ket{0}_{\frac{n}{2}+1}$$

Now, if $n \equiv 0 \mod 4$, then $\fl=\fr$, and we
can apply the method of Lemma~\ref{lemma:3} to un-compute 
the qubits $w_2,w_3, w_{\fl+1}$ using $\lceil \frac{n}{8} \rceil$ queries.

The step of Lemma~\ref{lemma:3} is applied so that the 
variables in $x'$ are the ones that are stored
in the work qubits as final states of those qubits.

If $n \equiv 2 \mod 4$ then $\fl=\fr-1$ and the state with
$w_1=\ket{1}$ requires $2$ less queries to obtain the related phase. It uses the two queries to transforms two qubits $w_{\fr}$ and $w_{\fr+1}$ to the same state as in the other
superposition state, in the same manner as shown for $f^{id}_6$
and thus after $\fl+2$ queries there are $\fr-2$ qubits that
need to be brought to the same state to un-entangle the system.
This takes a further $\lceil \frac{\fr-2}{2} \rceil$ queries 
using the methodology of Lemma~\ref{lemma:3}.
Thus in both cases, the algorithm requires $\lceil \frac{5n}{8}\rceil$ queries, and the system is in the state 

\begin{align*}
(-1)^{x_{i_1} + \ldots + x_{i_{\frac{n}{2}}} }& \Big( (-1)^{\phi_1{\hat{y}} \cdot \tilde{y}} \ketm{0}\ket{0}_1 \ket{x_{i_1}}_2\ldots \ket{x_{i_{\frac{n}{2}}}}_{\fr+1} \ldots \ket{0}_{\frac{n}{2}+1}
\\
+& (-1)^{\phi_1{\hat{z}} \cdot \tilde{z}} \ketm{0}\ket{1}_1 \ket{x_{i_1}}_2\ldots \ket{x_{i_{\frac{n}{2}}}}_{\fr+1} \ket{0}_{\fr+2} \ldots \ket{0}_{\frac{n}{2}+1} \Big)
\end{align*}

Finally in both cases the Hadamard gate is applied on the 
qubit $w_1$ which now contains the value of the permutation
$\phi(\hat{x}).\tilde{x}$ corresponding to the input given to the oracle.

At this point the work qubits $w_2$ through $w_{\fr+1}$ store the variables in $x'$, which are then used to calculate the 
value of the function $g$, XOR-ing the output of $g$ 
with $w_1$ and then measuring $w_1$ in the computational basis
gives us the final output.

\end{proof}
We call the set of MM Bent functions satisfying
the constraints of Theorem~\ref{th:4} as $\Gamma_n$.

\subsection*{The case of odd $n$}
So far, we have concentrated on the class of MM Bent functions, which are defined for all even $n$, and have obtained a 
large class of functions with 
deterministic query complexity of $n$ which our exact quantum algorithm 
evaluates using $\lceil \frac{5n}{8}\rceil$ queries.

However this technique can be extended for all odd values 
of odd $n$ as well. This can be done as follows.
\begin{enumerate}
\item Take any function on $f= \phi(\hat{x}).\tilde{x} \op g(x')$ on $n=2k$ variables such that $\phi$ and $g$ 
follow the constraints of Theorem ~\ref{th:4}. 

\item Form the function $f'=f(x) \op x_{n+1}$
\end{enumerate}
Since $f$ has a polynomial degree of $n$, as shown in \cite{parity}, this implies $f'$ has a polynomial degree of $n+1$.
This function can be evaluated in the exact quantum model by first evaluating $f$ using
$\lceil \frac{5n}{8} \rceil$ queries and using one more 
query to obtain the value of $x_{n+1}$.
Thus this takes $\lceil \frac{5n}{8} \rceil +1 \leq \lceil \frac{5(n+1)}{8} \rceil +1$ queries.
The number of functions that can be evaluated in this case
is same as that for $n$.

\subsection{The number of functions evaluated:}
\label{mm:num}
We finally calculate the number of functions covered via the definition of Theorem~\ref{th:3} for even $n$ ($\abs{\Gamma_n}$), and the number of functions 
for any odd $n$ is the same as the number of functions for
$n-1$.
We essentially give a lower bound on the number of functions, as our calculation is based on a single partition of $\hat{x}$ and $\tilde{x}$ into these four sets, and any choice of $x'$.

There are $2^{\fl}$ inputs to the first permutation and $2^{\fr}$ inputs to the second permutation,
and $x'$ contains $\fr$ inputs.
Therefore the total number of functions are $\left(2^{\fl}!\right)\left(2^{\fr}!\right)\left(2^{2^{\fr}}\right)$.

We now recall the definition of PNP-equivalence from~\cite{exact}.
\begin{definition}
Two functions $f$ and $g$ are called PNP-equivalent 
if $f$ can be obtained from $g$ by permuting the name of 
the variables in $g$, replacing some variables $x_i$ with $x_i \op 1$ in $g$
and by finally complementing the new formed function with $1$.
\end{definition}
If two functions are PNP equivalent then they have the same
deterministic and exact quantum query algorithm and often
an algorithm to evaluate one of them can be very easily modified
to evaluate the other using the same number of queries.

Corresponding to a function on $n$ variables, there can be at most $n! 2^{n+1}$ functions that are PNP-equivalent to it.
This is because there can be $n!$ permutation
of variables and each variable $x_i$ can be replaced with $x_i \op 1$, and finally each function $f(x)$ can be replaced with
$f(x) \op 1$. 
Also, the PNP-equivalence relation 
is reflective, symmetric and transitive in nature.

Therefore if there is a set of cardinality $\sf S$ consisting of functions on $n$ variables, then it consists of at least
$\frac{\sf S}{n!2^{n+1}}$ functions that are not PNP-equivalent.

Therefore in this case the class $\Gamma_n$
(exactly evaluated by our algorithm using $\lceil \frac{5n}{8} \rceil $ or $\lceil \frac{5n}{8} \rceil +1$
queries) must consist of at least
$$\frac{\left(2^{\fl}!\right)\left(2^{\fr}!\right)\left(2^{2^{\fr}}\right)}{n!2^{n+1}} =  \Omega 
\left(2^{\left(\fl 2^{\left( \fl \right)} \right)} \right) $$
functions, which is doubly exponential in $\fl$.

In conclusion, the fact that this algorithm cannot evaluate all MM Bent functions and thus all functions derived using 
the Bent concatenation method for odd values of $n$
is a limitation compared to the parity decision method, which we note down in the following remark.
\begin{remark}
\label{r:2}
The parity decision tree method in \cite{parity} evaluates all MM Bent functions on $n$ variables using $\lceil \frac{3n}{4} \rceil$ queries where as the algorithm described in this requires $\lceil \frac{5n}{8} \rceil$ queries, but is able to evaluate only the MM Bent functions that meet the constraints described in Theorem~\ref{th:4}.
\end{remark}

While the family of algorithms designed by us evaluates a class of functions super exponential in $\fl$,
with a query complexity lower than any known parity decision tree technique, it lacks in two areas.
The first is that we are unable to show that $\QC(f)=Q_E(f)$ for these functions.
The second is that we are unable to show $\QC(f) < D^2_{\op}(f)$ for any of these functions. 
That is, we do not know if there exists a parity decision tree technique that can have the same 
query complexity as the family of algorithms we have presented.
We have noted down in Theorem~\ref{th:par2} that $D_{\op}(f)$ is lower bounded by 
granularity. It is known that MM type Bent functions have a flat Fourier Spectra, 
with $\hat{f}(S)=\frac{1}{2^{\frac{n}{2}}} ~\forall~ S \subseteq [n]$.
Therefore granularity of any MM type Bent function is $\frac{n}{2}$
which gives us a lower bound that we can show to be tight. 

\section{Conclusion and Future Directions} 
\label{sec:5}
In this paper we have designed a new family of exact quantum algorithms ($\Q$) for certain classes 
of non-symmetric functions $f$ with query complexity $\QC(f)$.

First we have described the class ${\sf pdsp}(n,\frt,q)$ using perfect direct sum constructions 
with products, and shown that for a set of $\Omega(2^{\frac{\sqrt{n}}{2}})$ functions in this class 
we get $Q_E(f) = \QC(f) = \flt$ with $D_{\op}(f) > \flt$. For these set of functions we have 
$\flt+1 \leq D_{\op}(f) \leq n-1$, depending on the value of $q$ in ${\sf pdsp}(n, \frt, q)$.
We have obtained this result by designing exact quantum query algorithms based on $\F$ polynomial 
structure and then proven separation from generalized parity complexity technique 
by exploiting the high granularity of these functions. 

In this regard we design a subroutine as described in Theorem~\ref{lemma:0} which un-entangles two qubits in 
an entangled system with a single query, which allows us to obtain the said separations and is central to 
our algorithms. It would be interesting to study if this subroutine can be modified to be more efficient in the
bounded error quantum query model. 

In fact, we not only obtain advantage over the parity decision tree model in which the parity of two bits is 
calculated in a single query, but also the stronger generalized parity decision tree model in which parity of 
any number of bits can be calculated in a single query.

Using similar $\F$ polynomial based techniques we have also designed algorithms for a subclass of MM type Bent functions
(a variable XOR-ed with MM Bent function when $n$ is odd) consisting of at least $\Omega(2^{2^{\fl}})$ functions that
are not PNP equivalent for any value of $n$. This family of algorithms have query complexity of $\lceil \frac{5n}{8} \rceil$
where as the lowest query complexity of any known parity decision tree technique is $\frt$.
While $\Q(f)$ is optimal for $f= x_1x_4 \op x_2x_5 \op x_3x_6$,
we could neither show $\QC(f)=Q_E(f)$ or that $\QC(f)< \DP(f)$
for these classes of functions, which we note down here as open problems.
\begin{enumerate}
\item Does there exist any parity based method that can evaluate functions from this subclass using 
less than $\lceil \frac{3n}{4} \rceil$ queries?
\item What is the exact quantum query complexity of the functions in this class?
\end{enumerate}
Thus we design a family of algorithms that is both more powerful than the parity decision tree technique, 
and even the generalized parity decision tree technique and can be applied to a large class of non symmetric 
functions. In comparison, almost all the existing exact quantum query algorithms can only be applied to $poly(n)$ 
(mostly) symmetric functions. 
   
It remains of interest to understand the extent to which these techniques can be applied and how can they be modified 
to get optimal query complexity for other classes of Boolean functions, towards better understanding of this domain.


\begin{thebibliography}{}

\bibitem{amb1}
A. Ambainis.
Superlinear advantage for exact quantum algorithms.
In Proceedings of the forty-fifth annual ACM symposium on Theory of Computing (STOC'13),
891-900 (2013).
Arxiv: \url{https://arxiv.org/abs/1211.0721}

\bibitem{amb2}
A. Ambainis, J. Iraids and D. Nagaj.
Exact Quantum Query Complexity of $\sf EXACT^n_{k,l}$.
SOFSEM 2017: Theory and Practice of Computer Science, 243-255 (2016).

\bibitem{amb3}
A, Ambainis, J. Iraids and J. Smotrovs.
Exact quantum query complexity of EXACT and THRESHOLD.
Proceedings of the 8th conference on the theory of quantum computation, communication, and cryptography (TQC'13),
pp 263–269. 
Arxiv: \url{https://arxiv.org/abs/1302.1235}

\bibitem{amb4}
A. Ambainis, K. Balodis, A. Belovs, T. Lee, M. Santha and J. Smotrovs.
Separations in Query Complexity Based on Pointer Functions.
Journal of the ACM, DOI: \url{https://doi.org/10.1145/3106234} (2017).
Arxiv Version: https://arxiv.org/abs/1506.04719 (2015).

\bibitem{sdp}
H. Barnum, M. Saks and M. Szegedy.
Quantum query complexity and semi-definite programming.
In proceedings of 18th IEEE Annual Conference on Computational Complexity, pp. 179-193, (2003).

\bibitem{polypoly}
Robert Beals, Harry Buhrman, Richard Cleve, Michele Mosca and Ronald de Wolf.
Quantum lower bounds by polynomials. 
J. ACM 48, 4 (July 2001), 778–797. DOI: \url{https://doi.org/10.1145/502090.502097}.

\bibitem{pbook}
T. Cusick and P. Stanica.
Cryptographic Boolean Functions and Applications.
Academic Press, Elsevier (2009).

\bibitem{deutsch}
D. Deutsch  and R. Jozsa.
Rapid solution of problems by quantum computation.
In Proceedings of Royal Society London, vol. 439, issue 1907, pp. 553–558, 
DOI: \url{https://doi.org/10.1098/rspa.1992.0167}, 1992.

\bibitem{bent1}
J. F. Dillon. 
Elementary Hadamard Difference sets.
Ph.D. Dissertation, Univ. of Maryland (1974).

\bibitem{cvx}
M. Grant and S. Boyd.
CVH: Matlab software for disciplined convex programming, version 1.21, (2011).
\url{http://cvxr.com/cvx}

\bibitem{grover}
L. K. Grover. 
A fast quantum mechanical algorithm for database search.
In Proceedings of the twenty-eighth annual ACM symposium on Theory of Computing (STOC ’96),
Association for Computing Machinery, New York, NY, USA, 212–219 (1996).

\bibitem{exact}
A. Montanaro, R. Jozsa and G. Mitchison.
On Exact Quantum Query Complexity.
Algorithmica 71, 775–796 (2015).

\bibitem{parity}
C. S. Mukherjee and S. Maitra.
Classical-Quantum Separations in Certain Classes of Boolean Functions-- Analysis using the Parity Decision Trees.
Arxiv: \url{https://arxiv.org/abs/2004.12942} (2020).

\bibitem{self}
C. S. Mukherjee and S. Maitra.
Exact Quantum Query Algorithms Outperforming Parity -- Beyond The Symmetric functions (Extended Version)
Arxiv: \url{https://arxiv.org/abs/2008.06317v4}


\bibitem{nc}
M. A. Nielsen and I. L. Chuang.
Quantum Computation and Quantum Information. 10th Anniversary Edition,
Cambridge University Press, January 2011.

\bibitem{PAR} 
V. V. Podolskii and A. Chistopolskaya,
Parity Decision Tree Complexity is Greater Than Granularity, 
\url{https://arxiv.org/abs/1810.08668} (2018).

\bibitem{simon}
D. R. Simon.
On the Power of Quantum Computation.
SIAM Journal on Computing, vol. 26, no. 5, pp. 1474–1483, October 1997,
DOI: \url{https://doi.org/10.1137/S0097539796298637}.

\bibitem{shor}
P. W. Shor.
Polynomial-Time Algorithms for Prime Factorization and Discrete Logarithms on a Quantum Computer.
SIAM J. Comput. 26, 5, 1484–1509 (1997).

\end{thebibliography}
\end{document}